\documentclass[10pt, journal]{IEEEtran}

\usepackage{times}
\usepackage{epsfig}
\usepackage{amsfonts}
\usepackage{amsmath}
\usepackage{amsthm}
\usepackage{amssymb}
\usepackage{algorithm}
\usepackage{algorithmic}
\usepackage{multirow}
\usepackage{graphicx}
\usepackage{subfigure}
\usepackage{url}

\newcommand{\abs}[1]{\lvert #1 \rvert}

\newcommand{\card}[1]{\abs{#1}}
\newcommand{\norm}[1]{{\|{#1}\|}}
\newcommand{\normsq}[1]{{\|{#1}\|}^2}
\newcommand{\ind}[1]{{\mathbb I}_{\{#1\}}}
\newcommand{\wedef}{\stackrel{\triangle}{=}}

\newcommand{\smfrac}[2]{{\textstyle{\frac{#1}{#2}}}}

\newcommand{\mU}{{\mathcal{U}}}

\newcommand{\mP}{{\mathcal{P}}}
\newcommand{\mC}{{\mathcal{C}}}

\newcommand{\mJ}{{\mathcal{J}}}
\newcommand{\bp}{{\mathbf{p}}}

\newcommand{\bx}{{\mathbf{x}}}
\newcommand{\by}{{\mathbf{y}}}
\newcommand{\bz}{{\mathbf{z}}}

\newcommand{\mL}{{\mathcal{L}}}
\newcommand{\Ex}{{\mathbb{E}}}

\newtheorem{lemma}{Lemma}[section]
\newtheorem{example}{Example}

\newtheorem{proposition}{Proposition}[section]

\newtheorem{theorem}{Theorem}

\newtheorem{remark}{Remark}

\newcommand{\sinr}{\mathsf{SINR}}

\renewcommand{\baselinestretch}{0.95}

\begin{document}

\title{Learning Based Uplink Interference Management in 4G LTE Cellular Systems}
\author{Supratim Deb, Pantelis Monogioudis
\thanks{The authors are with
Wireless Chief Technology Office, Alcatel-Lucent USA.
e-mail: {\it first\_name.last\_name@alcatel-lucent.com}}}

\maketitle

\begin{abstract}

LTE's uplink (UL) efficiency critically depends on how the interference across
different cells is controlled. The unique characteristics of LTE's modulation
and UL resource assignment  poses considerable challenges in achieving this goal
because most LTE deployments have 1:1 frequency re-use, and the uplink
interference can vary considerably across successive time slots. In this
work, we propose LeAP, a measurement data-driven machine learning paradigm for
power control to manage uplink interference in LTE. The data driven approach has
the inherent advantage that the solution adapts based on network traffic,
propagation and network topology, that is increasingly heterogeneous with
multiple cell-overlays. LeAP system design consists of the following
components: {\em (i)} design of {\em user equipment} (UE) measurement statistics that
are succinct, yet expressive enough to capture the network dynamics, and {\em
(ii)} design of two learning based algorithms that use the reported measurements
to set the power control parameters and optimize the network performance. LeAP
is standards compliant and can be implemented in centralized SON ({\em self
organized networking}) server resource (cloud).  We perform extensive
evaluations using radio network plans from real LTE network operational in a
major metro~area in United States.
Our results show that, compared to existing approaches, LeAP provides
$4.9\times$ gain in the $20^{th}\%-\text{tile}$ of user data rate,
$3.25\times$ gain in median data rate.

\end{abstract}

\section{Introduction}
\label{sec:intro}


LTE uplink (UL) consists of a single carrier frequency division multiple access
(SC-FDMA) technique that orthogonalizes different users transmissions in the same
cell, by explicit assignments of groups of DFT-precoded orthogonal subcarriers. This
is fundamentally different from 3G, where 
users interfered with each other over the carrier bandwidth
and advanced receivers, such as successive interference cancellers, were
employed to suppress same cell interference. Same cell interference is mitigated by
design in LTE, however, other cell interference has a very different structure
compared to 3G. This calls for new power control techniques for setting user
transmit power and managing uplink interference.
In this paper, we design and evaluate LeAP, a new system that uses
measurement data to set power control parameters for optimal uplink interference
management in LTE.


\begin{figure}[t]
\begin{center}
\includegraphics[height=2.1in,width=2.6in]{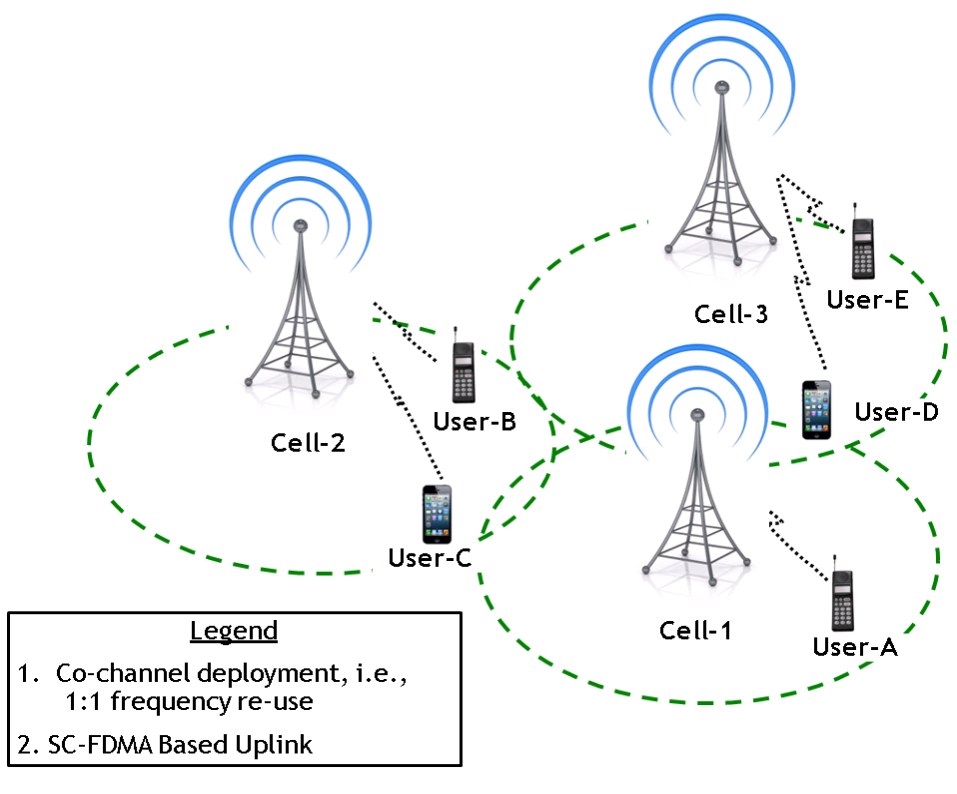}
\caption{\label{fig:toyeg}
{\small An illustration of three adjacent LTE cells. Depending on MAC time-frequency
assignment to users in Cell-2,~3, uplink transmission of User-A in Cell-1 could receive high interference from Cell-2
and low interference from Cell-3 or vice versa.}}
\end{center}
\end{figure}

\vspace{-0.1in}
\subsection{Uplink interference in 4G systems: Distinctive Properties}

Uplink interference in cellular systems is managed through careful power control
which has been a topic
of extensive research for more than two decades (see Section~\ref{sec:rw}). However,
uplink interference in LTE networks 
needs to be managed over multiple narrow bands (each corresponding to collection of
a {\em few} sub-carriers) over the entire bandwidth, thus giving rise to unique research  challenges.

To understand this better, we start by making two observations. Firstly, the uplink data rate of a
user in LTE depends on the SINR over the resource blocks (RBs)
\footnote{\scriptsize An RB is a block of
12 sub-carriers and 7 OFDM symbols and is the smallest allocatable resource in the
frequency-time domain.} assigned to the transmission.
Secondly, LTE uses 1:1 frequency re-use and thus an RB assigned to a user in a
cell can be used by another user in any of the neighboring cells too. The following example
illustrates how uplink interference is impacted due to the above observations.

{\bf Example:} Consider
the system in Figure~\ref{fig:toyeg} with three cells. Suppose the MAC of Cell-1 assigns,
to an associated User-A, RBs corresponding to time-frequency tuples
$(ts_1,fb_1), (ts_2,fb_2)$ over consecutive time slots $ts_1, ts_2$ respectively. Since the MAC of
Cell-2 and Cell-3 operate independently and use the same carrier frequency due to 1:1
re-use, the MAC
assignments of time-frequency tuples $(ts_i,fb_j),\ i,j=1,2$ at Cell-2 and Cell-3
could potentially be the following:
\begin{list}
{\labelitemi}
{\itemsep=0pt \parskip=0pt \parsep=0pt \topsep=0pt \leftmargin=0.2in}
\item Cell-2: 
$(ts_1,fb_1) \rightarrow \text{User-C}$ and $(ts_2,fb_2) \rightarrow \text{User-B}$
\item Cell-3: 
$(ts_1,fb_1) \rightarrow \text{User-E}$ and $(ts_2,fb_2) \rightarrow \text{User-D}$
\end{list}
As shown in the figure, since User-C and User-D are close to the edge of Cell-1,
in $ts_1$, User-A's transmission receives high interference from Cell-2
and low interference from Cell-3, whereas, in $ts_2$, User-A's transmission receives
high interference from Cell-3 and low interference from Cell-2. In any arbitrary
time-frequency resource,
interfering signal to User-A's uplink transmission at Cell-1, denoted by
$\text{Int}_A(LTE)$, can be
expressed as
\[ {\text{Int}}_A(LTE) = P_B\pmb{1}_{B}  + P_C\pmb{1}_{C} + P_D\pmb{1}_{D} +
P_E\pmb{1}_{E}\ ,  \]
where, $P_U, U\in\{B,C,D,E\}$ denotes the {\em received} power at Cell-1's base station due
to potential  uplink transmission of User-U in the {{\em same frequency block as
User-A}, and,  $\pmb{1}_U, U\in\{B,C,D,E\}$ is an $0\--1$ indicator variable denoting
whether User-U also transmits over the same time-frequency block as User-A.  Note
that, since User-B and User-C share the same cell, only one of them can be active in
a time-frequency resource and thus $\pmb{1}_B +\pmb{1}_C\leq 1$; similarly
$\pmb{1}_D+\pmb{1}_E\leq 1$. {Since MAC of each neighboring cell makes {\em independent scheduling
decision} on who gets scheduled in a time-frequency block,  the interference becomes highly
unpredictable from transmission to transmission; managing this unstable interference pattern
poses unique research challenges barely addressed in the literature. Indeed, this is unlike 3G
systems \footnote{\scriptsize In CDMA systems interference is summed over several
simultaneous user transmissions over the entire carrier bandwidth, thus leading to
more stable interference pattern.}, where the {\em neighboring cell
interference} for a similar topology
with CDMA technology (over an appropriate CDMA channel) would simply be
\[\text{Int}_A(CDMA)=P_B+P_C+P_D+P_E \] thus
leading to a more stable interference pattern across transmissions.  In general,
unlike LTE, the overall CDMA uplink interference also has an additional term for
{\em self cell interference} due to uplink users in the same cell.  Of course,
for desirable user performance, the uplink interference still has  to be managed
through power control algorithms that has been the focus of much of the existing
research on 3G power control and uplink interference management.

{\bf Solution requirements:} LTE networks are deployed with {\em self organized
networking} (SON) capabilities~\cite{ltesonsite1,ltesonsite2} to maximize network
performances.  Today's LTE networks are also
heterogeneous (HTNs) that include high power macro cells overlaying low power small
(pico/femto) cells.  Small cells are deployed in traffic hotspots or coverage
challenged areas, and thus, typically small cells outnumber macros by an order of
magnitude. This leads to a much larger, and hence more difficult to tune and manage,
cellular network where centralized SON servers are deployed to continuously optimize
the network~\cite{ltesonsite1}. Thus, a good solution to uplink interference
management should satisfy the following requirements: {\em (i)} it should be {\em
adaptive} to network traffic, propagation geometry and topology, {\em (ii)} it should
scale with the size of the network which could consist of tens of thousands of cells,
and {\em (iii)} it should be architecturally {\em compliant}  in the sense that it is
implementable in a SON server and adherent to standards. Note that, SON implementability dictates that the solution makes use of
the large amount of network measurement data. In this paper, we design a solution
that satisfy these requirements and achieves high network performance.

\vspace{-0.15in}
\subsection{Our Contributions}

In this paper, we propose LeAP, a learning based adaptive power control for uplink
interference management in LTE systems. We make the following contributions:

\begin{list}{\labelitemi}{\itemsep=0pt \parskip=0pt \parsep=0pt \topsep=0pt
\leftmargin=0.22in}

\item[1. ] {\em New framework for measurement data driven uplink interference
management:}
We propose a measurement data driven framework for setting power control parameters
for optimal uplink interference management
in an LTE network. Our framework {\em (i)} models the unique interference pattern in
OFDMA based LTE systems, {\em (ii)} accounts for varying traffic load and diverse propagation map in
different cells, {\em (iii)} and is implementable in a centralized SON architecture. 
See Section~\ref{sec:mod}\--\ref{sec:meas}.

\item[2. ] {\em Design of Measurement Statistics:} We derive
suitable measurement statistics  based on the processing of raw data from UE
measurement reports. Our measurement statistics are succinct yet expressive enough to
optimize the uplink performance by accounting for LTE's unique uplink interference
patterns along with network state and parameters. See Section~\ref{sec:meas} for
details.

\item[3. ] {\em Design of learning based algorithms:} Using the measurements derived
in Section~\ref{sec:meas}, we propose two learning based algorithms for optimal
setting of cell-specific power control parameters in  LTE networks. The two
algorithms trade-off complexity and performance: one provably converges to
the optimal, and the other is a fast heuristic that can be implemented using off-the-shelf solvers.  See
Section~\ref{sec:algospd}\--\ref{sec:algoreg}.

\item[4. ] {\em Extensive evaluation of LeAP benefits:} We evaluate our design using a
radio network plan from a real LTE network deployed in a major US metro. We
demonstrate the substantial gains for the evaluated network: the
edge data rate ($20^{th}\%-\text{tile}$ of data rate) improves to
$4.9\times$ whereas the median gain improves to $3.25\times$ compared to existing approaches.
The details are in Section~\ref{sec:eval}.

\end{list}



\vspace{-0.2in}
\section{Related Work}
\label{sec:rw}


\begin{sloppypar}
{\em Cellular Power Control:} Uplink power control in cellular systems has been an
active research area for around three decades. The pioneering works
in~\cite{yates-PC,foschini-PC} developed principles and iterative algorithms to
achieve target SINR when multiple users simultaneously transmit over a shared carrier.  This model is applicable to CDMA (3G) systems uplinks.
Since then, several authors have developed algorithms to jointly
optimize rate and transmit power in similar multi-user
systems~\cite{Saraydar01pricingand,Chiang04balancingsupply,xinzhou-PC,Xiao03autility-based}.
In particular, \cite{Saraydar01pricingand,xinzhou-PC,Xiao03autility-based} consider
utility based framework for joint optimization of rate and transmit power.  Uplink power control optimization was made tractable in CDMA setting in \cite{BocheS08,logconvSIR} where 
log-convexity of feasible SINR region was shown. We refer the reader
to \cite{book_powercontrol} for an excellent survey of the vast body of research in
power control. However, the uplink model in LTE is fundamentally different from these
systems for two reasons. First, unlike CDMA, uplink interference in LTE
is only from neighboring cells and the interference over an assigned {\em
frequency block} could change in every transmission, leading to a far more variable
interference pattern. Secondly, to control neighboring cell interference, the standards
have mandated cell-specific power control parameters that in turn govern UE
SINR-targets.
\end{sloppypar}

{\em Fractional Power Control (FPC) in LTE:} LTE power control is FPC based which
has led to some recent work~\cite{fpcericc,ibpclte,howtofpc,castfpc08}. However, unlike our work, none of
previous publications develop a framework and associated algorithms to optimize the FPC parameters based on
user path loss statistics and traffic load. Recognizing the
difficulties of setting FPC parameters in LTE,~\cite{sampatholpc} develops and
evaluates {\em closed loop} power control algorithms for dynamically adjusting SINR
targets so as to achieve a fixed or given interference target at every cell. This
work has two drawbacks: first, it is unclear how the interference targets could be
dynamically set based on network and traffic dynamics, and second, the scheme
does not maximize any ``network-wide" SON objective.

{\em Self Organized Networking (SON) in LTE:} Study of LTE SON algorithms 
have gained some attention recently mostly for downlink related
problems. \cite{twcRahmanY10,MadanBSBKJ10} study
the problem of downlink inter-cell interference coordination (ICIC)
for LTE, \cite{ffrStolyar08} optimizes downlink transmit
power profiles in different frequency carriers, and~\cite{ltedllb,ulblte,scsolte,slbhetnet}
study various forms (downlink, uplink, mobility based etc.) of traffic-load balancing with LTE SON.
To our best knowledge, ours is the first work to
develop measurement data-driven SON algorithms for uplink power control in LTE. For a
very extensive collection of material and presentations on the latest industry
developments in LTE SON, we refer the reader to~\cite{ltesonsite1,ltesonsite2}.



\vspace{-0.1in}
\section{A Primer on LTE Uplink and Fractional Power Control}
\label{sec:primer}

{\bf Terminologies used in the paper:}
A {\em resource block} (henceforth RB) is a block of 12~sub-carriers and 7 OFDM
symbols and is the smallest allocatable resource in the frequency-time domain.  {\em
eNodeB} (eNB in short) refers to an LTE base station and it hosts critical protocol
layers like PHY, MAC, and Radio Link Control etc. {\em User equipment}
(henceforth UE) refers to mobile terminal or user end device, and we also use UL for uplink.
Finally, {\em reference signal received power} (RSRP) is the average received power
of all downlink reference signals across the entire bandwidth
as measured by a UE. RSRP is a measure of downlink signal strength at UE.

\vspace{-0.1in}
\subsection{UL Transmission in LTE} 

LTE uplink uses SC-FDMA multiple access scheme. SC-FDMA reduces mobile's
peak-to-average power ratio by performing an $M$-point DFT precoding
of an otherwise OFDMA transmission. $M$ depends on the number of RBs assigned to the
UE. Also, each RB assigned to an UE is mapped to adjacent sub-carriers through
suitable frequency hopping mechanisms~\cite{3Gevol}.
At the base station receiver, to mitigate frequency selective fading, the per antenna signals
are combined in a frequency domain MMSE combiner/equalizer before the $M$-point
Inverse DFT (IDFT) and decoding stages are performed. The details of the above steps
are not relevant for our purpose; instead, we note  two key properties that will be
useful from an interference management perspective.
\begin{list}{\labelitemi}{\itemsep=0pt \parskip=0pt \parsep=0pt \topsep=0pt
\leftmargin=0.30in}
\item[{\bf P1:}]  The equalization of the received signal is performed
separately for each RB leading to {\em SINR being averaged only across
sub-carriers of a RB. Thus,  SINR in an RB  is a direct measure of UE
performance}\footnote{ The final UE data rate also depends
on the selected modulation and coding for the specific scheduled HARQ process.}.

\item[{\bf P2:}] Consider the adjacent subcarriers $S$ assigned  to an RB for a UE in a Cell-1.
In another neighboring Cell-2, the same subcarriers $S$ can be assigned to
at most 1~UE's RB because RB boundaries are aligned across subcarriers. 
Put simply, at a time there is at most 1 interfering UE per neighboring cell per RB.
\end{list}

\vspace{-0.1in}
\subsection{LTE Power Control and Challenges} 

To mitigate uplink interference from other cells and yet provide the UEs
with flexibility to make use of good channel conditions, LTE
standards have proposed that power control should happen at two time-scales as
follows:
\begin{list}{\labelitemi}{\itemsep=0pt \parskip=0pt \parsep=0pt \topsep=0pt
\leftmargin=0.22in}

\item[1. ] {\em Fractional Power Control (FPC-$\alpha$) at slow time-scale:} At the
{\em slower} time-scale (order of minutes), each cell sets
cell-specific parameters that the associated UEs use to set their average transmit
power and a target-SINR as a pre-defined function of local path loss measurements.
The two cell-specific parameters are nominal UE transmit power $P^{(0)}$ and fractional
path loss compensation factor $\alpha<1$.  A UE-$u$ with an average  path loss PL 
\footnote{\scriptsize $\text{Received Power at eNB}
 = \frac{\text{UE Transmit Power}}{\text{PL}}$}
to its serving cell transmits at a power spectral density (PSD),
\begin{equation}
P^{Tx}_{(u)}= P^{(0)}\cdot ({\text {PL}})^\alpha\ .
\label{eqn:fpcal}
\end{equation}
where PSD $P^{Tx}_{(u)}$ is expressed in Watt/RB. Thus, if a UE is assigned $M$
RBs in the uplink scheduling grant then its total transmit power is
$\min(MP^{Tx}_{(u)},P_{\text{tot}}^{\max})$ where $P_{\text{tot}}^{\max}$ is a cap on
the total transmission power of a UE. We remark that, $\alpha$ can be interpreted as
a {\em fairness} parameter that leads to higher SINR for UEs closer to eNB.

\item[2. ] At the {\em faster} time-scale, each UE is closed loop power controlled
around a mean transmit power PSD~(\ref{eqn:fpcal}) to ensure that a suitable average
SINR-target is achieved.  This closed loop
power control involves sending explicit power control adjustments, via the UL grants
transmitted in Physical Downlink Control Channel (PDCCH), that can be either
absolute or relative. 

\end{list}


{\bf Challenges:} LTE standards leave unspecified how each cell sets the value of
$P^{(0)}, \alpha$ and average cell-specific ``mean" interference targets that are useful
for setting UE SINR targets (see Remark~\ref{rem:sinrtarget},
Section~\ref{sec:meas}). Since the choice of parameters results in a suitable ``mean"
interference level at every cell above the noise floor, this problem of setting the
cell power control parameters is referred to as IoT control problem.  Clearly,
aggressive (conservative) parameter setting in a cell will improve (degrade) the
performance in that cell but will cause high (low) interference in neighboring cells.
Since FPC-$\alpha$ scheme and its parameters are cell-specific, its configuration
lies within the scope of {\em self-organized} (SON) framework, i.e., the
the solution should adapt the parameters based on suitable
periodic network measurements.  


\begin{remark}[Extensions]
The main idea behind FPC-$\alpha$ is to have cell dependent power control
parameters that change slowly over time. Our techniques
apply to any scheme where UE transmit power is a function of cell-specific
parameters and local UE measurements (path loss, downlink SINR etc.). 
\end{remark}

\vspace{-0.10in}
\section{System Model and Computation Architecture}
\label{sec:mod}

\begin{table}[t] \scriptsize
    \caption{List of Parameters and Variables used}
    \label{tab:params}
\begin{center}
\begin{tabular}{| c | c |}
    \hline
   {\bf Notation} & {\bf Description}  \\
\hline
\hline
    $\mC,\ c$     & Set of cell's and index for a \\
    $(c\in\mC)$   & typical cell, respectively \\
\hline
$P_c^{(0)}, \alpha_c$ & FPC-parameters of cell-$c$: nominal transmit power and \\
 & path loss compensation factor respectively. \\
\hline
$\pi_c$ & $\ln(P_c^{(0)})$, .i.e., nominal transmit power in log-scale \\
\hline
    $\mU_c$     & Set of UEs associated with cell-$c$ \\
\hline
    ${\mJ}_c $ & Set of cell's that interfere with cell-$c$ \\
\hline
    $l_c^{(u)}\ $ & Mean path loss of UE-$u$ to eNB of its serving cell-$c$ \\ 
\hline
    $L_c $ & Random variable for path loss of a\\ 
           & {\em random} UE $u\in \mU_e$ to its serving cell-$c$ \\
\hline
    $l_{e \rightarrow c}^{(u)}$ & Mean path loss of UE-$u\in \mU_e$ to cell-$c$'s eNB\\
\hline
    $L_{e \rightarrow c}$ & Random variable for path loss of a {\em random}\\ 
                          & UE-$u\in \mU_e$ to cell-$c$'s eNB \\
\hline
$\Lambda_c$, $\Lambda_{e \rightarrow c}$ & $\Lambda_c=\ln L_c$, $\Lambda_{e
\rightarrow c}=\ln L_{e\rightarrow c}$\\
\hline
$l_c(b),p_c(b)$ & For path loss histogram bin-$b$ of cell-$c$,\\
&  bin value and bin probability respectively\\
\hline 
$\gamma_c(b)$ & Expected SINR (in $\ln$-scale) of a cell-$c$ UE with path loss
$l_c(b)$.  \\
\hline
\end{tabular}
\vspace{-0.15in}
\end{center}
\end{table}

{\bf LTE HetNet (HTN) Model:} Our system model consists of a network of heterogeneous cells
all of which share the same carrier frequency , i.e., there is 1:1 frequency re-use in
the network. Some cells are high transmit power (typically 40~W) macro-cells while
others are low transmit power pico-cells (typically 1\--5~W). In HTNs, pico cells\footnote{\scriptsize Picocells as opposed
to femto cells have open subscriber group policies.}  are deployed by operators in
traffic hotspot locations or in locations with poor macro-coverage. The distinction
between macro and pico is not relevant for the development of our techniques and
algorithms but it is very important from an evaluation of our design. See
Section~\ref{sec:eval} for further discussions and insights. $\mC$ denotes
the set of all cells that includes macro and pico cells. We will introduce parameters
and variables as needed; Table~\ref{tab:params} lists the important ones.

{\bf Interfering Cell:} For a typical cell indexed by
$c\in \mC$, the set of other cells
that interfere with cell-$c$ are denoted by ${\mJ}_c$.
Ideally, a cell $e\in \mJ_c$, if the uplink
transmission of some UE associated with cell-$e$ is received at eNB of cell-$c$ with
received power above the noise floor. {\em Since interfering cells must be defined based
on the available measurements, we say $e\in\mJ_c$ if some UE-$u$ associated to
cell-$e$ reports downlink RSRP (i.e. RSRP is above a threshold of 140~dBm) from cell-$c$, 
which is an indication that UE-$u$'s uplink signal could interfere with
eNB of cell-$c$.} This is reasonable assuming uplink and downlink path-loss symmetry between
UEs and eNB.

\begin{figure}[t]
{\small
\begin{center}
\includegraphics[height=1.4in,width=2.85in]{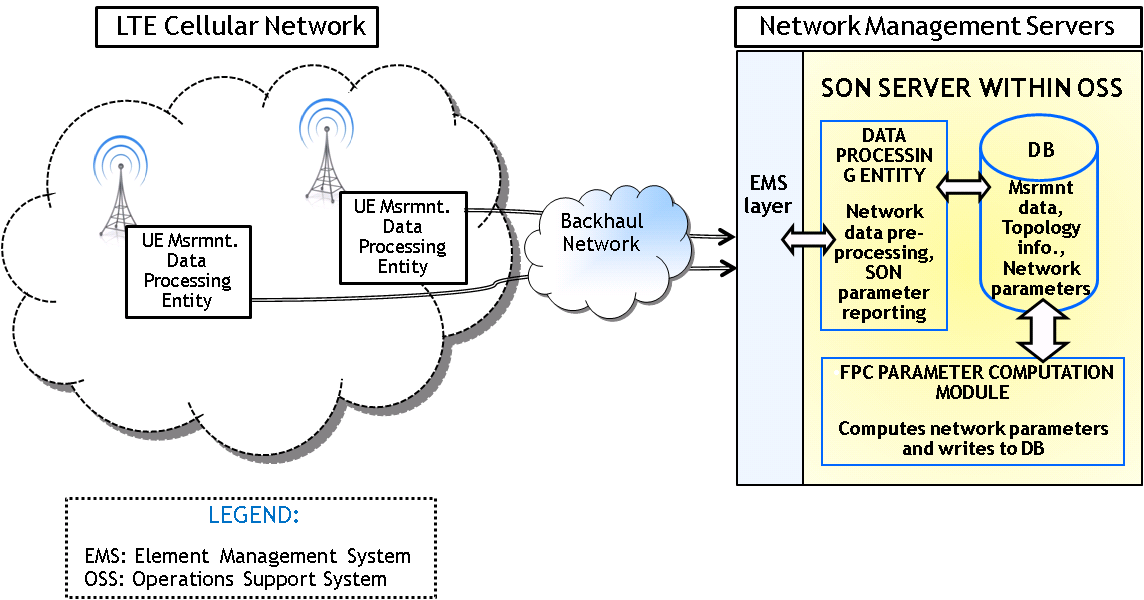}
\vspace{-0.1in}
\caption{\label{fig:comparch}
{\small Underlying computation architecture.}}
\end{center}
}
\vspace{-0.33in}
\end{figure}

{\bf Underlying Self-Optimized Networking (SON) architecture:} The underlying SON architecture
is shown in Figure~\ref{fig:comparch}. In this
architecture, network monitoring happens in a distributed manner across
the radio access network, and the heavy duty algorithmic computations happen
centrally at the Network Management Servers. 
As opposed to a fully distributed (across eNBs) computation approach,
this kind of
{\em hybrid} SON architecture is the preferred by
many operators for complex SON use-cases~\cite{ltesonsite1,ltesonsite2} for two
key reasons: capability to work across base stations from
different vendors as is typically the case, and not having to deal with convergence
issues of distributed schemes (due to asynchrony and message latency). We note two
relevant aspects of this architecture:

\begin{list}{\labelitemi}{\itemsep=0pt \parskip=0pt \parsep=0pt \topsep=0pt
\leftmargin=0.18in}

\item {\em Main building blocks:} The main building blocks are the following: {\em
(i)} monitoring
components at the cells that  collect network measurements,
appropriately process them to create Key Performance Indicators
(KPIs), and communicate the measurement and KPIs to the central server, and {\em (ii)}
the algorithmic computation engine (cloud servers) at a
central server that makes use of the KPIs and
compute the SON parameters that are then fed back to the network.

\item {\em Time-scale of computation:} The KPIs from the network are typically
communicated periodically. The period should be such that, it is short enough to
capture the changing network dynamics that call for network re-configuration and is
long enough for accurate estimation the relevant statistics. In general, most
networks have periodic measurement reports with frequency
$5\--15$~mins~\cite{ltehandbook} but
much faster {\em minimization of drive test} (MDT) data and UE trace information can be collected. 
In addition, measurement reports can be trigerred by events like
traffic load above a certain threshold beyond normal.

\end{list}

{\bf Design questions:} Given this architecture, to periodically compute the
FPC-$\alpha$ parameters to maximize the network performance in the uplink, we need to answer two
questions:

\begin{list}{\labelitemi}{\itemsep=0pt \parskip=0pt \parsep=0pt \topsep=0pt
\leftmargin=0.18in}
\item{\bf Q1:} What are the minimal set of network measurements required to configure FPC-$\alpha$
parameters? 

\item  {\bf Q2:} Based on these measurements, how should we choose $\alpha$,
$P^{(0)}$ and average interference threshold for every cell ? The algorithms should
be scalable and capable of updating the parameters as new measurement data arrives.

\end{list}

In Section~\ref{sec:meas}, we answer Q1 and propose a measurement based optimization framework , and in
Section~\ref{sec:algospd}\--\ref{sec:algoreg} , we develop learning based
algorithms to answer Q2.

\vspace{-0.12in}
\begin{sloppypar}
\section{Measurement and Optimization Framework}
\label{sec:meas}
\end{sloppypar}

\vspace{-0.05in}
\subsection{Network Measurement Data}

{\bf Additional notations:} We will derive an expression of SINR of a typical uplink
UE. Towards this goal, consider a network snapshot with a collection of UEs $\mU$; a
generic UE is indexed by $u\in \mU$.  Let $\mU_c$ denote the set of UEs
associated with cell-$c$. 
Denote by $l_c^{(u)}$ path loss from $u\in\mU_c$ to its
serving cell  $c$.  We also drop the superscript for UE and write $L_c$ to denote
the path loss of a {\em random} UE associated with cell-$c$.  We denote the {\em mean}
path loss from a UE $u\in \mU_e$ to an cell $c\in \mJ_{e}$ by $l_{e \rightarrow
c}^{(u)}$. We also drop
the superscript for UE and write $L_{e\rightarrow c}$ to denote the path loss from
cell-$e$ to cell-$c$ of a {\em random} UE belonging to cell-$e$. 
These notations are also shown in Table~\ref{tab:params}.

\begin{remark}[Fast-fading and frequency selectivity]
All path loss variables must be interpreted in ``time-average" sense so that the effect
of fast-fading is averaged out. Indeed, fast-fading happens at a much faster
time-scale (ms) than the FPC-$\alpha$ parameter computation time-scale of minutes.
Also, SC-FDMA equalization at the receiver  averages out the effect of
{\em frequency selectivity} in an RB.
\end{remark}

{\bf Expression for SINR:} 
Consider a UE-$u$ that transmits to its serving cell-$c$ over
the RBs assigned to it. Consider the transmission over any one
RB that is assigned subcarriers in the set $S$ by the MAC scheduler.
To derive the SINR over $S$, we wish to quantify the interference
experienced by the received signal at serving cell-$c$ over $S$. Due to our
observations on frequency hopping, the subcarriers within $S$ can be used in a
neighboring cell by at most one UE. Note that, depending on the load on a
cell, a RB may only be utilized only during a certain fraction of the
transmission subframes. With this motivation, we define the following binary
random variable:
\begin{equation}
O_{e\rightarrow c} \wedef \left\{
\begin{array}{lc}
1 & 
\begin{array}{c}
\text{\small if cell-$e$ schedules a UE over the }\\
\text{\small resource block $S$ also used at cell-$c$}
\end{array}\\
0 & \text{\small else}
\end{array}\ \right. 
\label{eqn:defoccupy}
\end{equation}
In the above, we assume that $O_{e\rightarrow c}$  does not depend on the
specific choice of $S$ and  has identical distribution for all $S$.
Denoting by $v$ the UE that occupies $S$ at cell-$e$,
ignoring fast-fading, 
the interfering signal at serving cell-$c$ of UE-$u$ for transmissions over
$S$ in cell-$e$ can be expressed as follows.
\begin{align*}
I_{e\rightarrow c} 
& = O_{e\rightarrow c} P^{Tx}_{(v)}\ .\  l_{e\rightarrow c}^{(v)}
= O_{e\rightarrow c} P^{(0)}_e (l_e^{(v)})^{\alpha_e} (l_{e\rightarrow c}^{(v)})^{-1} 
\end{align*}
It is instructive to note that, even without
fluctuations due to fast fading, the interference $I_{e\rightarrow c}$ is random because the interfering
UE-$v$ transmitting over $S$ in cell-$e$ could be any random UE in cell-$c$.
{\em In other words, the quantities $l_e^{(v)}, l_{e\rightarrow
c}^{(v)}$ can be viewed as a random sample from the joint distribution of
$(L_e,L_{e\rightarrow c})$.} Thus the total interference at cell-$c$ over frequency
block $S$ is a random variable $I_c$ given by

\begin{equation}
\label{eqn:totalint}
I_c = \sum_{e\in \mJ}
O_{e\rightarrow c} P^{(0)}_e (L_e)^{\alpha_e} (L_{e\rightarrow c})^{-1}
\end{equation}
where we drop the superscript of the loss variables to mean that the loss is from a
{\em random} UE in an interfering cell using the RB-$S$. This randomness is induced
by MAC scheduling of the interfering cell. Now, the SINR of UE-$u$ in cell-$c$ denoted
by $\sinr(l_c^{(u)})$ (as a function of UE-$u$'s path loss $l_c^{(u)}$) over
RB-$S$ can be expressed as

\begin{align}
& \sinr(l_c^{(u)}) = \frac{P^{Rx}_{(u)}}{I_c + N_0} \nonumber \\
& =\frac{P^{(0)}_{c} (l_c^{(u)})^{-(1-\alpha_c)}}
{\sum_{e\in \mJ_c}
O_{e\rightarrow c} P^{(0)}_e (L_e)^{\alpha} (L_{e\rightarrow c})^{-1}+N_0}
\label{eqn:sinrfb}
\end{align}

\begin{sloppypar}
{\bf Measurement variables:} It follows from~(\ref{eqn:sinrfb}) that, for any given
$l_c^{(u)}$, $\sinr(l_c^{(u)})$ is a random variable that is fully characterized by the
following distributions: joint distribution of $(L_e,L_{e\rightarrow c})$ and
$O_{e\rightarrow c}$ for all $e\in \mJ_c$. Furthermore, as we argue formally in
Section~\ref{sec:iotcprob},
{\em towards computing an ``average" network wide performance metric over the SON
computation period, we also require $\rho_c$ the mean uplink load in cell-$c$
and path loss distribution $L_c$.} The dependence on traffic load on the average network
performance is also quite intuitive. We thus have the following. 

\vspace{-0.05in}
\begin{center}
\fbox{\parbox[c]{0.95\linewidth}{
{\bf Design of Measurement Statistics:} 
For a given set of values of $P^{(0)}_c,\alpha_c$, under FPC-$\alpha$ mechanism, any
expected network wide performance metric can be fully characterized by the following
statistics: 
\begin{list}{\labelitemi}{\itemsep=0pt \parskip=0pt \parsep=0pt \topsep=0pt
\leftmargin=0.18in}

\item[1. ] 
Joint path loss 
distribution $(L_e,L_{e\rightarrow c})$ of uplink UEs for every interfering cell $e$ of $c$.

\item[2. ]
$\Pr(O_{e\rightarrow c}=1)$ which is the probability that cell-$e$ schedules a UE interfering 
with cell-$c$ for transmission over an RB.

\item[3. ]
Mean number of active  uplink UEs (uplink load) in every cell-$c$ denoted by $\rho_c$.

\item[4. ]
Path loss distribution of uplink UEs of cell-$c$ denoted by $L_c$.

\end{list}
}}
\end{center}
\end{sloppypar}


{\bf Histogram construction:} The required distributions can be
estimated using standard histogram inference techniques~\cite{BishopML} using the following
steps we state for completeness:
\vspace{0.05in}



\noindent{1. \em Collecting UE measurement samples:} 
Since UE measurement reports\footnote{The measurements are either available through
what is called {\em per call measurement data} or it can be seeked by
eNB\cite{TS36.331}.} contain reference signal strength (termed
RSRP) from multiple neighboring cells, the RSRP values can be coverted into path losses
using knowledge of cell transmit power. Thus, for a UE associated with cell-$c$, if a
measurement report contains RSRP from cell-$c$ and cell-$e$ both, then it 
provides samples for $L_c, $ and $(L_e,L_{e\rightarrow c})$.

\noindent {2. \em Binning the path loss data:} This is a standard step where
the range of $L_c$ and $(L_e,L_{e\rightarrow c})$ are divided into several
disjoint histogram bins and each data point is binned appropriately.
\vspace{0.05in}

\noindent {3. \em Estimating the occupancy probabilities $\Pr(O_{e\rightarrow c}=1)$}:
Assuming proportional-fair MAC scheduling\footnote{\scriptsize This is the most
prevalent MAC scheduling
in LTE.} where all UEs in a cell use the radio resources uniformly,
$\Pr(O_{e\rightarrow c}=1)$ can be estimated as the fraction of UEs in cell-$e$
that interfere with cell-$c$. This estimation can be performed by simply
computing the fraction of UEs in cell-$e$ that report measurements from cell-$c$.
\vspace{0.05in}

We make two observations. First, the histograms are best maintained in dB scale to
so that the range of path loss values is not too large.
Second, in practice, the data samples for constructing the
histograms can come for a large enough random subset of all the data.

\vspace{-0.1in}
\subsection{The IoT Control Problem in LTE HTNs}
\label{sec:iotcprob}

We are now formulate our problem based on the measurement histograms
described in the previous section.

{\bf Network wide performance metric:} Given our model, a suitable performance metric
should satisfy two properties: {(i)} it should account for the randomness in the
SINR's in a meaningful manner, and {(ii)} it should strike a balance between
aggregate cell throughputs and fairness.

We wish to propose a average performance metric where average is over all UE path
losses that can realize the measurement data.
Towards this end, we first define a performance metric for a typical
UE-$u$ in cell-$c$ who has path loss to its serving cell given by $l_c^{(u)}$.
Denote by $\gamma_c(l_c^{(u)})$ as the expected log-SINR of a typical UE-$u$ in
cell-$c$, i.e., $\gamma_c(l_c^{(u)})=\Ex[\ln \sinr(l_c^{(u)})]$ where the expression
of $\sinr(l_c^{(u)})$ is given by~(\ref{eqn:sinrfb}). We now define the
UE utility function $V(.)$ as follows:
\begin{equation}
V(\gamma_c(l_c^{(u)})) \wedef
\left\{
\begin{array}{c} 
\text{\small Utility of a UE-$u$ in cell-$c$ as a }\\
\text{\small function of $\gamma_c(l_c^{(u)})=\Ex[\ln \sinr(l_c^{(u)})]$}
\end{array}
\right\} 
\label{eqn:perUEutil}
\end{equation}
Here $V(.)$ is a concave increasing function and
we have explicitly shown the SINR depends on the path loss $l_c^{(u)}$.
Choosing the utility as a function of
$\Ex[\ln \sinr_u]$ has two benefits. Firstly, $\ln(\sinr_u)$ is a measure of the data rate with
$\sinr_u$. Secondly, converting the SINR into a log-scale makes the problem
tractable since it is well known that feasible power region is
log-convex~\cite{logconvSIR}\footnote{\scriptsize In simple terms, if all powers are represented
in log-scale then the feasible power vectors of all UEs is a convex region.}, and
thus, we can use elements of convex optimization theory.

Having defined a utility for a typical UE, we are now in a position to define
network wide performance metric obtained by averaging 
over all UE path losses that can realize the measurement data.
We need additional notations for the histogram of $L_c$'s.  
Suppose the path loss histogram of $L_c$ is
divided into $k$ disjoint intervals with mid-point of the intervals given by $l_c(1),
l_c(2),\hdots, l_c(k_c)$. Suppose in cell-$c$, the empirical probability\footnote{\scriptsize The
empirical probability of histogram bin-$b$ is simply the number of data items
binned into the bin-$b$ divided by the total number binned items for this
histogram.} of items in histogram bin-$b$ is given by $p_c(b)$. 
Then, we define the overall system utility as the {\em expected utility
of all UEs in all the cells} where the expectation is over empirical path loss
distribution given by the measurement data. This utility, as a function of vector of
$P^{(0)}_c,\alpha_c$'s (denoted by ${\mathbf P^{(0)}},\pmb{\alpha}$) is given by.
{\allowdisplaybreaks
{\small
\begin{align}
&\text{Util}({\mathbf P^{(0)}},\pmb{\alpha}) 
=\sum_{c\in \mC} \Ex[\text{Total utility of UEs in cell-$c$}]
\nonumber \\
&=\sum_{c\in \mC} \sum_{b=1}^{k_c}\Ex[\text{Total utility of UE with path loss
$l_c(b)$ in cell-$c$}]
\nonumber \\
&=\sum_{c\in \mC} \sum_{b=1}^{k_c}\Ex[\text{{\small Num. UEs with path loss 
$l_c(b)$}}] \times V(\gamma_c(l_c(b)))
\nonumber \\
& = \sum_{c\in \mC} \sum_{b=1}^{k_c}\rho_c\  p_c(b)\ V\left(\gamma_c(l_c(b))\right)
\label{eqn:sysutil}
\end{align}
}
}
where $\gamma_c(l_c(b))$ is the expected SINR in log-scale for any UE with path
loss $l_c(b)$ to its serving eNB of cell-$c$. The last step follows because the expected
number of UEs with path loss $l_c(b)$ is product of expected number of UEs in
cell-$c$ and the probability of a UE with path loss $l_c(b)$.

{\bf Choice of UE-utility:} Though our techniques work for a generic concave and
increasing $V(.)$, for our design and evaluation, we 
use the following form of $V(.)$ in the rest of paper:
\begin{equation}
\label{eqn:samutil}
V(\gamma) = \ln \ln(1+\exp(\gamma))\ .
\end{equation}
\begin{sloppypar}
One can verify that the function $V(.)$ defined above is concave.  
Roughly speaking, $\ln(1+\exp(\gamma))$ is the Shannon data rate corresponding to
log-scale (natural log) SINR of $\gamma$. It is well known that utility
defined by log of data rate strikes the right balance between {\em fairness} and
overall system performance.

\end{sloppypar}

\begin{remark}[Data rate and throughput]
In this paper, we use UE data rate or PHY data rate (in bits/sec/Hz) instead of UE
throughputs (bits/sec/Hz). Under proportional-fair MAC scheduling, the dominant scheduling policy in LTE,
UE throughput $R_u$ of UE-$u$ in cell-$c$, is roughly related to PHY data rate $r_u$
of UE-$u$ by $R_u=r_u/N_c$; here $N_c$ is the average number of simultaneously active users in
cell-$c$. Thus, $\ln R_u=\ln r_u - \ln N_c$. In other words, the
log of throughput and log of data rate are off by an additive constant independent of
power control parameters. 
\end{remark}

We also observe that optimizing the total log-data rate of all
UEs is aligned with the MAC objective of proportional-fair scheduling in LTE.

\begin{sloppypar}
{\bf Measurement based IoT Control Problem (IoTC):} We now formulate the problem of optimally
configuring the parameters of the FPC-$\alpha$. This problem is
commonly referred to as the Interference over Thermal (IoT) control problem
because the optimal choice of $P^{(0)}_c,\alpha_c$'s result in a suitable interference
threshold above the noise floor in every cell. The problem can be succinctly stated
as follows: 
\begin{quote}
{\em Given the path loss histograms of $L_c, (L_e,L_{e\rightarrow c})$,
empirical distribution $O_{e\rightarrow c}$, and average traffic load $\rho_c$ at every cell-$c$, 
find optimal $P^{(0)}_c,\alpha_c$ for every cell-$c$ such that we maximize $\text{Util}({\mathbf
P^{(0)}},\pmb{\alpha}) $ given by~(\ref{eqn:sysutil}).}
\end{quote}
\end{sloppypar}

At a first glance, the above problem may look complex due to the fact that
$\text{Util}({\mathbf P^{(0)}},\pmb{\alpha}) $ is a function of several random
variables.  We will first formulate this problem as a mathematical program so that
the problem becomes tractable. By using motivation from~\cite{logconvSIR}\footnote{\scriptsize
\cite{logconvSIR} shows that feasible power region is log-convex 
in many wireless networks. This means that, if powers are in log-domain, then
the problem has convex feasible region which makes it amenable to convex optimization
theory.}, we convert all powers, and path losses to logarithmic scale as
follows.

{\allowdisplaybreaks {\small \begin{align} & \theta_c=\Ex\left[\ln \left( {\sum_{e\in
\mJ_c} O_{e\rightarrow c} P^{(0)}_e (L_e)^{\alpha} (L_{e\rightarrow c})^{-1}+N_0}
\right)\right] \label{eqn:logtran1} \\ & \pi_c = \ln P^{(0)}_c  \  , \ \lambda_{c}(b)
= \ln l_c(b)  \ , \ \lambda_{ec}(b) = \ln l_{ec}(b)  \ , \label{eqn:logtran2} \\ 
& \ \Lambda_c = \ln L_c  \ , \ \Lambda_{e\rightarrow c} = \ln L_{e\rightarrow c}
\label{eqn:logtran3} \end{align} } } 

\begin{sloppypar}
The notations are also shown in Table~\ref{tab:params}.
We can now re-write $\gamma_c(l_c(b))$ (given by~(\ref{eqn:sinrfb})) in terms
of these above variables as
\begin{align}
& \gamma_c(l_c(b))  = \pi_c-(1-\alpha_c)\lambda_c(b) - \theta_c
\label{eqn:logtran4}\\
&\text{where,}\nonumber \\
&\  \theta_c =\Ex\left[\ln 
\left(
\sum_{e\in \mJ_c}
O_{e\rightarrow c} e^{(\pi_e +\alpha_e \Lambda_e - \Lambda_{e\rightarrow
c})} + N_0 
\right) \right]
\label{eqn:logtran5}
\end{align}
which has the standard look of SINR expression in log-scale. Thus, we can state the
problem of maximizing $\text{Util}({\mathbf
P^{(0)}},\pmb{\alpha}) $ as the problem of maximizing $$\sum_c \sum_b\rho_cp_c(b)
V(\gamma_c(l_c(b)))$$
subject to equality constraints given by~(\ref{eqn:logtran4})~and~(\ref{eqn:logtran5}).
It turns out that the equality constraints can be re-written as inequalities to
convert this into convex non-linear program (NLP) with inequality constraints.
This leads to the following proposition
where we have also imposed an upper bound on the maximum transmit power per RB and
maximum average interference cap.
\end{sloppypar}

\begin{proposition}
\begin{sloppypar}
Given measurement statistics of $L_c, (L_c,L_{e\rightarrow c}), O_c$ and average
traffic load $\rho_c$ at every cell-$c$, the problem of maximizing
$\text{Util}({\mathbf P^{(0)}},\pmb{\alpha}) $ given
by~(\ref{eqn:sysutil}), subject to maximum transmit power constraint on RB, 
is equivalent to the following {\em convex} program:
\end{sloppypar}
\begin{center}
\fbox{\parbox[c]{0.98\linewidth}{
\vspace{-0.1in}
{\small
\begin{align}
& \text{{\bf IoTC-SCP:}}  \nonumber \\
& \max_{\{\pi_c\},\{\alpha_c\},\{\gamma_c(b)\}} \ \  \sum_{c,b}\rho_cp_c(b)
V(\gamma_{c}(b))
\nonumber\\
& \text{subject to,} \nonumber \\
& \ \forall\ c\in \mC,b\in[1,k_c] :\ \gamma_{c}(b)  \leq \pi_c - (1-\alpha_c) \lambda_c(b) - \theta_c 
\label{eqn:con1}\\
& \ \forall\ c\in \mC:\  \theta_c  \geq \Ex\left[\ln 
\left(
\sum_{e\in \mJ_c}
O_{e\rightarrow c} e^{(\pi_e +\alpha_e \Lambda_e - \Lambda_{e\rightarrow
c})} + N_0 
\right) \right]
\label{eqn:con2} \\ 
&\ \forall\ c\in \mC,b\in[1,k_c] :\ \pi_c + \alpha_c \lambda_c(b)\leq \ln P_{max}
\label{eqn:con3}\\
&\ \alpha_c \in [0,1],\ \gamma_c(b)\in [\gamma_{min},\infty),\ \ \theta_c\in [\ln
N_0,\ln I_{max}]
\label{eqn:con4}
\end{align}
}
\vspace{-0.2in}
}}
\end{center}
\end{proposition}
\begin{proof}
(Outline) First, at optimality, all of
the inequalities must be met with equality from basic principles. Secondly, the
constraint set is a convex region since the constraints~(\ref{eqn:con2}) is convex 
due to convexity of LSE function\footnote{\scriptsize LSE or
log-sum-exponential functions are functions of the form $f(\{x_i\}_i)=\ln(\sum_i
e^{\sum x_{ij}a_j})$.}~\cite{ghaconv}.
\end{proof}

We make two important remarks.
\begin{remark}[IoTC-SCP constraints]
In IoTC-SCP, constraint~(\ref{eqn:con1}) states that the expected SINR in log-scale
can be no more than the receiver UE power less the expected uplink interference at
the eNB, constraint~(\ref{eqn:con2}) states that the expected interference at each
cell has to be at least the expected sum of interferences from the neighboring cells,
and constraint~(\ref{eqn:con3}) limits the maximum power allowed per RB.
Finally,~~(\ref{eqn:con4}) states the valid domain of the
variables $\alpha_c,\theta_c$'s and $\gamma_u$'s, where, $\gamma_{min}$ is the minimum
decodable SINR. Note that the valid domain of $\theta_c$ is $[\ln N_0,\ln I_{max}]$
because maximum average interference in cells could need a cap due to
hardware design constraints. 

\end{remark}
\begin{remark}[Setting SINR target of UEs]
\label{rem:sinrtarget}
In IoT-SCP, suppose $\pi_c^*,\alpha_c^*,\theta_c^*$ represent the optimal values
for cell-$c$ and let $P_c^{(0)*}=\exp(\pi_c^*), I_c^*=\exp(\theta_c^*)$ be the optimal
values in linear scale. Then, the SINR target of an arbitrary UE in cell-$c$ with
average path loss $l_c^{(u)}$ is set as
{\small
\begin{equation}
\label{eqn:sinrtarget}
\sinr_{\text{target}}(l_c^{(u)})
=\frac{\min\left[P_{\max}\ (l_c^{(u)})^{-1},\ P_c^{(0)*}\ 
(l_c^{(u)})^{-(1-\alpha^*_c)}\right]}{I_c^*}\ .
\end{equation}
}
\end{remark}

\begin{sloppypar}
{\bf Challenges in solving IoTC-SCP:} IoTC-SCP is
somewhat different from
a traditional non-linear program that are solved using standard gradient
based approaches (using Lagrangian).
The main challenge in solving IoTC-SCP comes from the randomness in the
interferers.  As a consequence of this, the expected value of r.h.s. of interference
constraint~(\ref{eqn:con2}) is computationally-intensive even for given
values of the $\pi_c$'s and $\alpha_c$'s. Indeed, this requires accounting for
$O(B_{max}^K)$ number of interferer combination per cell where $B_{max}$ is the maximum
number of histogram-bins of the joint distribution of $(L_e,L_{e \rightarrow c})$
in any cell and $K$ is the maximum number of interferes of an cell.
Repeating this computation for every iteration of a gradient based
algorithm is next to impossible.

By adapting elements of learning theory, we
propose two algorithms that strike a different performance-computation trade-off.
One is optimal requiring more computational resources and the
other is an approximate heuristic with fast computation times.
\end{sloppypar}

\vspace{-0.18in}
\begin{sloppypar}
\section{A Stochastic Learning Based Algorithm}
\label{sec:algospd}
\end{sloppypar}

We now develop a stochastic-learning based gradient algorithm. 

We introduce additional notations for ease of exposition. Denote the original
optimization variables $(\gamma_c(b),\pi_c,\alpha_c)$'s by the vector ${\mathbf z}$
(also called {\em primal} variables). We also denote the constraint set of the
problem IoTC-SCP by the random function
$h({\mathbf z},{\boldsymbol O, \Lambda})$, i.e.
\begin{align}
&h({\mathbf z},{\boldsymbol O, \Lambda})\nonumber \\
&=\left[
\begin{array}{c}
\left[\gamma_c(b)  - \pi_c - (1-\alpha_c)
\lambda_c(b)\right]_{c \in \mC, b\in [1,k_c]} \\
\left[\ln \left( \sum_{e\in \mJ_c}
O_{e\rightarrow c} e^{(\pi_e +\alpha_e \Lambda_e - \Lambda_{e\rightarrow
c})} + N_0 \right)  - \theta_c\right]_{c\in\mC} \\
\left[ \pi_c + \alpha_c
\lambda_c(b) - \ln P_{max}\right]_{c\in \mC,b\in [1,k_c]}
\end{array}
\right] 
\label{eqn:compcons}
\end{align}
Thus, the IoTC-SCP problem can be simply stated as
\[\max_{{\mathbf z}}\sum_{c,b} V(\gamma_c(b))\ \ 
\text{s.t.}\ \ \Ex[h({\mathbf z},{\boldsymbol O, \Lambda})]\leq 0\ .\]
In the above, the expectation is with respect to the random variables $O_{e\rightarrow c}$'s and
$(\Lambda_e,\Lambda_{e\rightarrow c})$'s.
We define the random Lagrangian function as follows:
\begin{align}
\label{eqn:lagra}
\mL({\boldsymbol z}, h({\mathbf z},{\boldsymbol O}, {\boldsymbol \Lambda}),{\mathbf p} )
=\sum_u V(\gamma_u) -  {\mathbf p}^t h({\mathbf z},{\mathbf O}, {\boldsymbol \Lambda})
\end{align}
where, ${\mathbf p}\geq 0$ denotes the so called Lagrange multiplier vector and it
has the same dimension as the number of constraints.  Since the objective function of
IoTC-SCP is concave and the constraint set is convex, one can readily show from
convex optimization theory that~\cite{convopt}
\begin{align}
\label{eqn:sp}
\text{OPT} &= 
\max_{{\mathbf z}} \min_{{\mathbf p}\geq 0} 
\Ex[\mL({\boldsymbol z}, h({\mathbf z},{\boldsymbol O, \Lambda}),{\mathbf p}
)]\\
&=\min_{{\mathbf p}\geq 0} \max_{{\mathbf z}}
\Ex[\mL({\boldsymbol z}, h({\mathbf z},{\boldsymbol O, \Lambda}),{\mathbf p} )]\ ,
\nonumber
\end{align}
where OPT denotes the optimal value of IoTC-SCP. Our goal is to solve the IoTC-SCP
problem by tackling the saddle point problem~(\ref{eqn:sp}).

{\bf Useful bounds on the optimization variables:} Before we describe our  algorithm,
we derive trivial but useful bounds on feasible primal variables. The iterative
scheme to be described shortly projects the updates primal variables within these
bounds. The proof is straightforward and skipped for want of space.

\begin{lemma}
\label{LEM:PUB}
Any feasible $\gamma_c(b)$ and $\pi_c$ of the problem IoTC-SCP satisfies 
\[\pi_{\min} \leq \pi_c \leq  \pi_{\max},\ \ 
\gamma_{\min} \leq \gamma_c(b) \leq \gamma_{\max}\ , \]
where,
{\small \[\pi_{\max} = \ln P_{\max},\ \pi_{\min} = \gamma_{\min} + N_0,\ 
\gamma_{\max}= \ln P_{\max} - N_0\ .\]}
\end{lemma}

{\bf Intuition and a stochastic iterative algorithm:} 
The equivalent problem~(\ref{eqn:sp}) is referred to
as the {\em saddle point} problem and one could use standard {\em primal-dual}
techniques to solve this problem~\cite{asusaddlepoint}. In such an approach, 
the primal variables
and dual variables are updated alternatively in an iterative fashion, and each update
uses a gradient (of  $\Ex[\mL(.)]$) ascent for the primal variables and gradient
descent for the dual variables. However,
computing the expectation and its gradient is computationally expensive (can be
$O(B_{max}^K)$ where $B_{max}$ is maximum histogram bins and $K$ is maximum number
of interferers) and this has to be done in every iteration which is practically
impossible. Instead,
we perform primal-dual iterations by substituting the random variables with {\em random
samples} of the random variables in every iteration: thus expectation computation is
replaced by a
procedure to draw a random sample from the histograms which is a computationally
light procedure. The idea of replacing a random variable by a {\em sample} in an iterative
scheme is not new, the entire field of stochastic approximation theory deals with
such schemes and conditions for this to work~\cite{borkarstochapp}. Our contribution
is in applying this powerful technique to the IoTC-SCP problem. By using elements of
stochastic approximation theory, we also show that the resulting output converges to optimal. 

We now describe  our algorithm.
Denote the value of the primal variables in iteration-$n$ by
${\mathbf z}_n$ and the value of the dual variables in iteration-$n$ by ${\mathbf
p}_n$. We also denote by $\hat{{\mathbf z}}_n$ and $\hat{{\mathbf p}}_n$ the
iteration average
of the primal and dual variables respectively.
The overall algorithm is summarized in Algorithm~\ref{algo:iotspd}
and is described as follows:


\begin{sloppypar}
\noindent {\bf Initialization:} First ${\mathbf z}_n,{\mathbf p}_n,\hat{{\mathbf
z}}_n,\hat{{\mathbf p}}_n$ are initialized to positive values within bounds specified
by Lemma~\ref{LEM:PUB}.

\noindent {\bf Iterative steps:} Denote by $a_n=1/n^{\zeta},0.5<\zeta\leq 1$ as the
step size in iteration-$n$.  The following steps are repeated for $n=0,1,2,\hdots$:

\begin{list}{\labelitemi}{\itemsep=0pt \parskip=0pt \parsep=0pt \topsep=0pt
\leftmargin=0.22in}

\item[1.] {\em Random Sampling of Interferers:} This step is performed for each
interfering cell pair $(c,e)$. For each cell-$e\in\mJ_c$ the following steps are
performed:
\begin{list}{\labelitemi}{\itemsep=0pt \parskip=0pt \parsep=0pt \topsep=0pt
\leftmargin=0.16in}

\item Toss a coin
with probability of {\em head} $\Pr(O_{e\rightarrow c}=1)$.  Denote the outcome of
this coin-toss by the $0\--1$ variable $\chi_{ec}$ where $\chi_{ec}$ takes value one if
there is a head.  

\item If $\chi_{ec}=1$, then draw a random sample from the joint distribution of
$(L_e, L_{e\rightarrow c})$ based on the histogram of this joint distribution.
Denote the random sample, which is a 2-tuple by $(sa[1],sa[2])$. 
Let $(\xi_{ec}[1],\xi_{ec}[2])$ denote the logarithm of this random sample, i.e.,
$\xi_{ec}[i]=\ln(sa[i]),i=1,2\ .$

\item In the function $h(.)$ given by~(\ref{eqn:compcons}) representing the
constraints of IoTC-SCP, replace the random variables
$(\Lambda_e,\Lambda_{e\rightarrow c})$ by the random sample
$(\xi_ec[1],\xi_{ec}[2])$.

\end{list}
Denote by ${\boldsymbol \chi}, {\boldsymbol \xi}$ as the vector of
$\chi_{ec},\xi_{ec}$'s samples. We write $h({\mathbf z}_n, 
{\boldsymbol \chi}, {\boldsymbol \xi}),{\mathbf p}_n)$ as a function of these random
samples along with the primal and dual variables in iteration-$n$.

\item[2.] {\em Primal Update:} The primal variables are updated as
\begin{equation}
\label{eqn:pud}
{\mathbf z}_{n+1} = {\mathbf z}_n + a_n 
\nabla_{\bz}\mL({\mathbf  z}_n, h({\mathbf z}_n, 
{\boldsymbol \chi}, {\boldsymbol \xi}),{\mathbf p}_n )
\end{equation}
where $\nabla_{\bz}\mL(.)$ denotes the partial derivative of $\mL(.)$ with respect
to ${\mathbf z}$. 

$\ \ $Next project updated values of $\alpha_e,\pi_e,\theta_e,\gamma_u$'s respectively within
intervals $[0,1]$, $[\pi_{\min},\pi_{\max}]$, $[\ln N_0,\ln I_{max}]$,$[\gamma_{\min},\gamma_{\max}]$. For
example, if the value of $\alpha_e$ in iteration-$(n+1)$, $({\alpha_e})_{n+1}>1$,
then set $({\alpha_e})_{n+1}\leftarrow 1$; and, if any
$({\alpha_e})_{n+1} <0$ set $({\alpha_e})_{n+1}\leftarrow0$. 

\item[3.] {\em Dual update:} The dual variables are updated as
\begin{equation}
\label{eqn:dud}
{\mathbf p}_{n+1} = \left[ {\mathbf p}_n + 
a_n h({\mathbf z}_n, {\boldsymbol \chi},{\boldsymbol \xi})\right]^+\ .
\end{equation}

\item[4.] {\em Updating average iterates:} The current value of the solution is given by
average iterates
\begin{align*}
\hat{{\mathbf z}}_{n+1} & \leftarrow \smfrac{1}{n}{\mathbf z}_n + (1-\smfrac{1}{n})\hat{{\mathbf z}}_n\\
\hat{{\mathbf p}}_{n+1} & \leftarrow \smfrac{1}{n}{\mathbf p}_n + (1-\smfrac{1}{n})\hat{{\mathbf p}}_n\ .
\end{align*}

\end{list}

The iterative steps are repeated for $N_{iter}$ iterations based on whether
$\hat{{\mathbf z}}_{N_{iter}}, \hat{{\mathbf p}}_{N_{iter}}$ converge within a desirable
precision. There are many techniques are testing the convergence of stochastic
iterations~\cite{stop-simul} that can be readily used for our purpose. The step-sizes
$a_n$ can be set as $1/n^{\zeta},0.5<\zeta\leq 1$ or in an adaptive manner to speed
up the convergence~\cite{stepsize-stochapp}.

\end{sloppypar}

\begin{algorithm}[t]
\caption{\label{algo:iotspd} \textsc{IoTC-SL}:
Stochastic Primal Dual Algorithm for IoTC Control.}
{\small
\begin{algorithmic}[1]

\STATE Initialize the primal variables to ${\mathbf z}_{0}$ and the
dual variables to ${\mathbf p}_{0}.$

\FOR{iterations $n=1$ to $n=N_{it}$}

\STATE For each interfering cell pair $(e,c)$ draw a random path loss sample from the
joint distribution histogram of $(L_e,L_{e\rightarrow c})$ and the replace the path
loss random variables by these samples in the expression of $h(.)$
in~(\ref{eqn:compcons}).

\STATE Update the primal variables according to gradient ascent based update rule
in~(\ref{eqn:pud}).

\STATE Update the dual variables according to gradient descent based update rule
in~(\ref{eqn:dud}).

\STATE Maintain average of the primal and dual variables over all the iterations
\ENDFOR

\STATE The average of the primal variables over all the iterations is the output.

\end{algorithmic}
}
\end{algorithm}

\begin{sloppypar}
{\bf Asymptotic Optimality of Algorithm~IoTC-SL}
We will now state the main analytical result for Algorithm~IoTC-SL. The following
result shows that the iterates generated by
Algorithm~IoTC-SCP converge to the optimal solution for ``almost" every {\em
sample path} of the iterates.

\end{sloppypar}



\begin{theorem}
\label{THM:SPD}
Suppose ${\mathbf z}^*$ is the optimal solution of IoTC-SCP. Then, the following holds:
$${\mathbf z}_n\rightarrow {\mathbf z}^*\ \ \text{with probability}\ 1\ .$$
\end{theorem}
\begin{proof} See Appendix.
\end{proof}

\vspace{-0.15in}
\begin{sloppypar}
\section{Regression Based Certainty Equivalent Heuristic}
\label{sec:algoreg}
\end{sloppypar}

For ease of exposition, we re-write the random interference in
constraint~(\ref{eqn:con2}) of IoTC-SCP which is
{\small
\begin{equation}
\label{eqn:reintcon}
\forall\ c\in \mC:\  \theta_c  \geq \Ex\left[\ln \left( 
\sum_{e\in \mJ_c}
O_{e\rightarrow c} e^{(\pi_e +\alpha_e \Lambda_e - \Lambda_{e\rightarrow
c})} + N_0
\right)\right]\ .
\end{equation}
}
The basic idea of this heuristic consists of two high-level methods. First is the
notion of {\em
certainty equivalence} where we replace the random log-interference by log of
expected interference in~(\ref{eqn:reintcon})
such that any solution with this modified constraint is a feasible solution of
IoTC-SCP. Second is the notion of {\em regression}, where, using a known parametric
distribution as model for the path loss statistics, we perform a fitting of the
distribution parameters using the measurement data. Now the problem can be shown to
be a standard non-linear program (NLP) that can be solved using off-the-self NLP
solvers. We next explain the different steps. 

{\bf Certainty equivalent approximation:} We approximate the random interference
constraint~(\ref{eqn:reintcon}) as follows:
{\small
\begin{equation}
\label{eqn:ce1}
\forall\ c\in \mC:\  \theta_c  \geq \ln \left( 
\sum_{e\in \mJ_c}
\Ex\left[
O_{e\rightarrow c} e^{(\pi_e +\alpha_e \Lambda_e - \Lambda_{e\rightarrow
c})} + N_0
\right] \right)\ .
\end{equation}
}
We now make a very important observation.  Denote by $F_{approx}$ the feasible
solution region of the optimization problem obtained by the above
modification~(\ref{eqn:ce1}) of interference constraint~(\ref{eqn:con2})  in
IoTC-SCP.  If $F_{IoTC-SCP}$ is the feasible region of IoTC-SCP, then it can be shown
from Jensen's inequality that $(\ref{eqn:ce1})\Rightarrow (\ref{eqn:reintcon})$, and thus,
\[ F_{approx} \subseteq F_{IoTC-SCP}\ .\]
This means, if we replace constraint~(\ref{eqn:reintcon}) by~(\ref{eqn:ce1}), 
we will produce feasible solution to IoTC-SCP.

Denoting $a_{e\rightarrow c}=\Ex[O_{e\rightarrow c}]$,~(\ref{eqn:ce1}) can be rewritten as
\begin{equation}
\label{eqn:ce1a}
\theta_c \geq \ln\left(\sum_{e\in \mJ_c}a_{e\rightarrow c} g(P^{(0)}_e,\alpha_e) + N_0\right)\ ,
\end{equation}
where
\begin{equation}
\label{eqn:ce1b}
g(\pi_e,\alpha_e) = \Ex[e^{(\pi_e +\alpha_e \Lambda_e - 
\Lambda_{e\rightarrow c})}]\ .
\end{equation}

{\bf Model for path losses:} To simplify the r.h.s. of ~(\ref{eqn:ce1}), we
use the following widely used path-loss model:

\begin{sloppypar}
{\sc Distribution Model for Path loss:} {\em For each interfering cell pair, $(c,e)$, the joint distribution of the random
variables $(L_e, L_{e\rightarrow c})$ are log-normal. In other words, the
distribution of $(\Lambda_e, \Lambda_{e\rightarrow c})=(\ln L_e, \ln L_{e\rightarrow
c})$ is jointly Normal.}
\end{sloppypar}

Indeed, log-normal distribution is a very reasonable statistical model for spatial
path loss distribution~\cite{goldsmith-wirelss}. For ease of exposition, we define
the following two dimensional vectors: 

\begin{equation}
{\mathbf Z}_{ec} =
\left[ \begin{array}{c} \Lambda_e \\ \Lambda_{e\rightarrow c} \end{array}\right],\ \ 
\beta_e=
\left[ \begin{array}{c} \alpha_e \\ -1 \end{array}\right],\ \ 
\label{eqn:cezbeta}
\end{equation}
${\mathbf Z}_{ec}$ is defined for each interfering cell pair $(c,e)$ and $\beta_e$ for
each cell $e$. Note that, under our modeling assumption, the random variable ${\mathbf
Z}_{ec}$ is bi-variate jointly Normal random variable. The parameters of ${\mathbf
Z}_{ec}$ are defined by the following mean and correlation matrix:

\begin{equation}
{\mathbf m}_{ec} = \Ex[{\mathbf Z}_{ec}] \ ,\ 
{\mathbf C}_{ec} = 
\Ex[({\mathbf Z}_{ec}-{\mathbf m}_{e}) ({\mathbf Z}_{ec}-{\mathbf m}_{e})^t]
\label{eqn:cemc}
\end{equation}
 Note that ${\mathbf m}_{ec}$ is a $2\times 1$ vector and ${\mathbf C}_{ec}$ is a
$2\times 2$ symmetric positive definite matrix. Also, ${\mathbf m}_{ec},{\mathbf C}_{ec}$
are model parameters that can be estimated from the measurement data using standard techniques
for Gaussian parameter estimation~\cite{kay-estimation} (also outlined in
Algorithm~\ref{algo:iotcce}).

{\bf Model based simplification of IoTC-SCP:} 
Using the definition of ${\mathbf Z}_{ce}$ in~(\ref{eqn:cezbeta}), we can
rewrite $g(\pi_e,\alpha_e)$ in~(\ref{eqn:ce1b}) as follows.
\begin{equation*}
g(\pi_e,\alpha_e) = \Ex[\exp(\pi_e + \beta_e^t {\mathbf Z}_{ec})]
\end{equation*}
Since ${\mathbf Z}_{ec}$ is Normal, the above expression can be simplified using
moment generating functions of multivariate Normal random variables
as~\cite{kay-estimation}:
\begin{equation}
\hat{g}(\pi_e,\alpha_e) = \exp\left( \pi_e + \beta_e^t {\mathbf m}_{e} + 
\smfrac{1}{2}\beta_e^t {\mathbf C}_{ec}\beta_e \right)
\label{eqn:ceei}
\end{equation}
where we write $\hat{g}(.)$ to mean that it is an estimate of $g(.)$ under
the assumed statistical model of path losses. We can now solve the following
modified version of IoTC-SCP problem where we replace the interference
constraint~(\ref{eqn:con2}) by
\[ I_c \geq \ln\left( \sum_{e\in \mJ_c } a_{e\rightarrow c}\hat{g}(\pi_e,\alpha_e) + N_0\right) .\]
This is stated as follows:\\
\begin{center}
\fbox{\parbox[c]{0.95\linewidth}{
{\small
\begin{align*}
&\text{\bf IoTC-CE:}  \nonumber \\
&\max_{\{\pi_e\},\{\alpha_e\},\{\gamma_c(b)\}} \sum_{c,b}\rho_c p_c(b)V(\gamma_c(b)) \nonumber\\
&\text{subject to,} \nonumber\\
& \ \ \ \gamma_c(b)  \leq \pi_c - (1-\alpha_c) \lambda_c(b) -
\theta_c \nonumber\\
& \ \ \ \theta_c  \geq 
\ln\left[ 
\sum_{e\in \mJ_c}a_{e\rightarrow c}e^{( \pi_e + \beta_e^t {\mathbf m}_{ec} + 
\smfrac{1}{2}\beta_e^t {\mathbf C}_{ec}\beta_e)} + N_0
\right] \nonumber \\
& \ \ \ \pi_c + \alpha_c \lambda_c(b)\leq \ln
P_{max}\nonumber \\
& \ \ \  \alpha_c \in [0,1],\ \theta_c\in [\ln N_0,\ln I_{max}],\ \gamma_c(b)\in [\gamma_{min},\infty)
\end{align*}
}
\vspace{-0.1in}
}}
\end{center}

\begin{algorithm}[t]
\caption{\label{algo:iotcce} \textsc{SolveIoTC-CE}:
Regression based deterministic heuristic for IoTC COntrol.}
{\small
\begin{algorithmic}[1]

\STATE {\em Regression step:} For each interfering cell pair $(e,c)$, we
use the path-loss measurements to find a normal fit of the Gaussian random
random variable ${\mathbf Z}_{ec}$ given by~(\ref{eqn:cezbeta}).
Suppose the histogram bins based on the measurement of joint distribution of
$(L_e,L_{e\rightarrow c})$ are indexed by $b$ where $b=1,2,\hdots$. Suppose
bin-$b$ represents the values $(l_e(b),l_{e\rightarrow c}(b))$ and has empirical
probability mass $p_{{ec}}(b)$. Let $(\lambda_e,\lambda_{e\rightarrow
c})=(\ln l_e(b),\ln l_{e\rightarrow c}(b))$.
Then a maximum-likelihood estimator~\cite{kay-estimation} of ${\mathbf m}_{ec},{\mathbf
C}_{ec}$ is as follows:
\begin{equation*} 
{\mathbf m}_{ec} =
\left[
\begin{array}{c}
\sum_b p_{ec}(b) \lambda_e(b)\\
\sum_b p_{ec}(b) \lambda_{e\rightarrow c}(b)
\end{array}
\right], \
{\mathbf C}_{ec} =
\left[ 
\begin{array}{cc}
\hat{\sigma_e^2}\ & \hat{\sigma_{ec}^2} \\
\hat{\sigma_{ec}^2} & \hat{\sigma_e^2}
\end{array}
\right]
\end{equation*}
where,
\begin{align*} 
\hat{\sigma_{e}^2} & =
\sum_b p_{ec}(b) (\lambda_e(b)-{\mathbf m}_{ec}[1])^2 \\
\hat{\sigma_{ec}^2} & =
\sum_b p_{ec}(b) (\lambda_e(b)-{\mathbf m}_{ec}[1])
(\lambda_{e\rightarrow c}(b)-{\mathbf m}_{ec}[2]) 
\end{align*}

\STATE {\em Deterministic non-linear program solving:}
With the above estimates, we solve the NLP given by IoTC-CE 
using standard tools for solving deterministic non-linear
programs (such as, dual based techniques, primal-dual techniques etc.)~\cite{bertNLP}. 

\end{algorithmic}
}
\end{algorithm}

IoTC-CE is convex deterministic NLP (and
hence standard tools of solving NLP can be used), provided the function
$\hat{g}(\pi_e,\alpha_e))$ is convex in $(\pi_e, \beta_e)$. The following lemma shows
shows that $\hat{g}(\pi_e,\alpha_e))$ is indeed a convex function.

\begin{lemma}
The function $\ln(\sum_{e\in \mJ_c}\hat{g}(\pi_e,\alpha_e)+N_0)$ is convex in $\pi_e$
and $\alpha_e$. Thus IoTC-CE is a convex NLP.
\end{lemma}
\begin{proof}
Since convex and component-wise increasing function of a vector
of convex functions is convex~\cite{convopt}, the result follows.  The details are
omitted for want of space.
\end{proof} 

{\bf Overall algorithm by putting it all together:} As outlined in Algorithm~\ref{algo:iotcce}, 
we have two steps. First we use the
measurement data to estimate the parameter-matrices
${\mathbf m}_e, {\mathbf C}_{ec}$ of the distribution ${\mathbf Z}_{ce}$. Second, we
use these values to solve IoTC-CE using standard NLP solving techniques.

The main benefit of the approach in this section is that, it is computationally not
so intensive and can be solved using powerful commercial NLP solvers; the price we pay is the sub-optimality. 


\vspace{-0.15in}
\section{Evaluation using Propagation Map from a Real LTE Deployment}
\label{sec:eval}

The primary goal of our evaluation is two fold. First, to understand the gains provided by
LeAP as compared to other popular approaches. Second, to understand the effect of
important design parameters like histogram bin-size and IoT-cap on LeAP performance.
In addition, we also provide a comparison of the two proposed LeAP algorithms.

\vspace{-0.1in}
\subsection{Evaluation Platform} 

{\bf Evaluation platform:}  For our evaluation purpose, we implemented the following
components in Figure~\ref{fig:comparch}: 
SQL database (DB) that saves network
information, the data access layer which interacts with the DB and creates the
measurement data discussed in Section~\ref{sec:meas},
and the IoT-control parameter
($P^{(0)},\alpha$'s) computation engine.  All our implementations were based on .NET
using C\# and are multi-threaded using the Task Parallelization Library (TPL) facilities. We
implemented Algorithm~\ref{algo:iotspd} within our framework using suitable data
structures to enable multi-threaded  implementation. We skip these implementation details 
for want of space.  For network measurement information saved in the DB, we used real signal
propagation map from an actual deployment but synthetic mobile locations generated
using a commercial network planning tool in a manner described shortly.

Note that, in our evaluation platform, we have not yet implemented
Algorithm~\ref{algo:iotcce}. Such implementation requires integrating a .NET
compatible commercial NLP solver within our architecture, a feature that is in
progress. However, just for the purpose of comparison, we have developed a
stand-alone implementation of Algorithm~\ref{algo:iotcce} using CVX~\cite{cvx} a free
NLP Matlab solver. This can operate only on moderate sized networks with restricted
objective function forms and is very slow for practical needs.

\vspace{-0.1in}
\subsection{Evaluation Setup and Methodology}

\begin{figure}[t]
{\small
\begin{center}
\includegraphics[height=1.4in,width=2.7in]{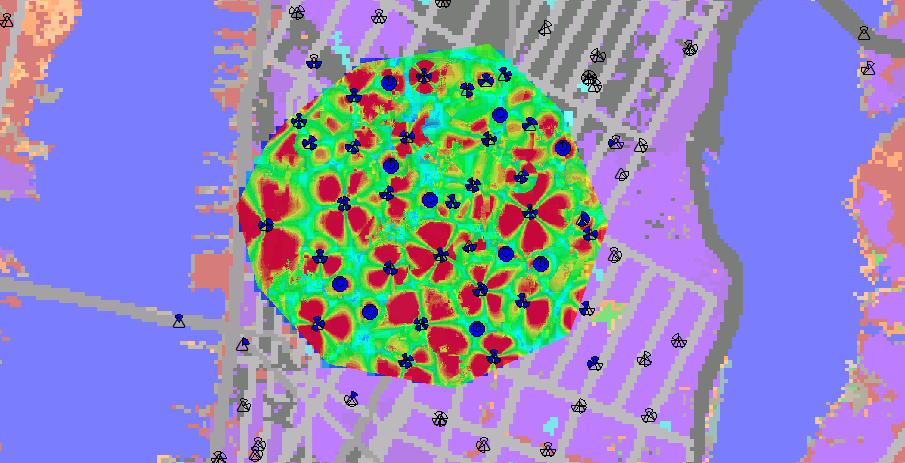}
\vspace{-0.1in}
\caption{\label{fig:nyctop}
{\small LTE signal map in evaluated section of a major US~metro. 
The heat map is based on downlink SINR as the uplink SINR depends on specific FPC
scheme.}}
\end{center}
}
\vspace{-0.3in}
\end{figure}

We used signals propagation maps generated by a real operational LTE network in
a major US~metro area\footnote{The name of the city and the operator cannot be
revealed for proprietary reasons.}
(see Figure~\ref{fig:nyctop} for the propagation heat map).  The terrain
category, cell site locations, and drive-test propagation data from the network were
fed into a commercially available Radio Network Planning (RNP) tool that is used by
operators for cellular planing~\cite{9955tool}. The carrier bandwidth is 10~MHz in
the 700~MHz LTE band.  There is no dependency of LeAP design to the specific tool
used for evaluation purpose.

For our evaluation, we selected an area of around 9~$\text{km}^2$ in the central
business district of the city with 115~macro cells. This part of the city has a very high density of macro
cells due to high volume of mobile data-traffic.  Going forward,
LTE deployments are going to have low-power pico cells in addition to high power
macro cells. Since picos are not yet deployed in reality, 10~pico locations were
manually embedded into the network planning tool using its built-in capabilities.
Macros have transmit power 40~W and picos have transmit power 4~W. All including, we
have 125~cells in the evaluated topology.

Though we had propagation data based on a real deployment, we do not have access to
actual mobile measurements at this stage. To generate measurement data required for
our evaluation, we generated synthetic mobile locations using the capabilities of RNP
tool. These location were used to generate measurement data as follows:

\begin{enumerate}
\item Using the drive-test calibrated data, terrain information and
statistical channel models, the RNP tool was used to generate signal propagation
matrix in every pixel in the area of interest in a major US~metro as shown in
Figure~\ref{fig:nyctop}.

\item The RNP tool was then used to drop thousands of UEs in
several locations where the density of dropped UEs was as per dense-urban
density (450 active mobile per sq-km). In addition, we defined mobile hotspots
around some of the picos where the mobile density was doubled. Based on the signal propagation matrix in
every pixel and mobile drop locations, the RNP tool readily generated all path loss data from
UEs to it serving cell and neighboring cells and also mobile to cell association matrix.

\item The mobile path loss data was used to generate the histograms and the other
measurement KPIs proposed in Section~\ref{sec:meas}.  This data was then fed into
LeAP database. 

\end{enumerate}

Once the measurement data is available, the results were generated as follows:

\begin{enumerate}
\item The measurement data from LeAP database was read by our implementation of
LeAP algorithm that generated the cell power control parameters. 

\item To evaluate the gains, the parameters generated by LeAP algorithm
(or any other comparative algorithm) were evaluated for a random UE snapshot
generated by RNP tool where each UE's data-rate based on its SINR-target given
by~(\ref{eqn:sinrtarget}) was computed.

\end{enumerate}

\begin{figure}[t]
\begin{center}
\subfigure[]{
\includegraphics[height=1.6in, width=3.2in]{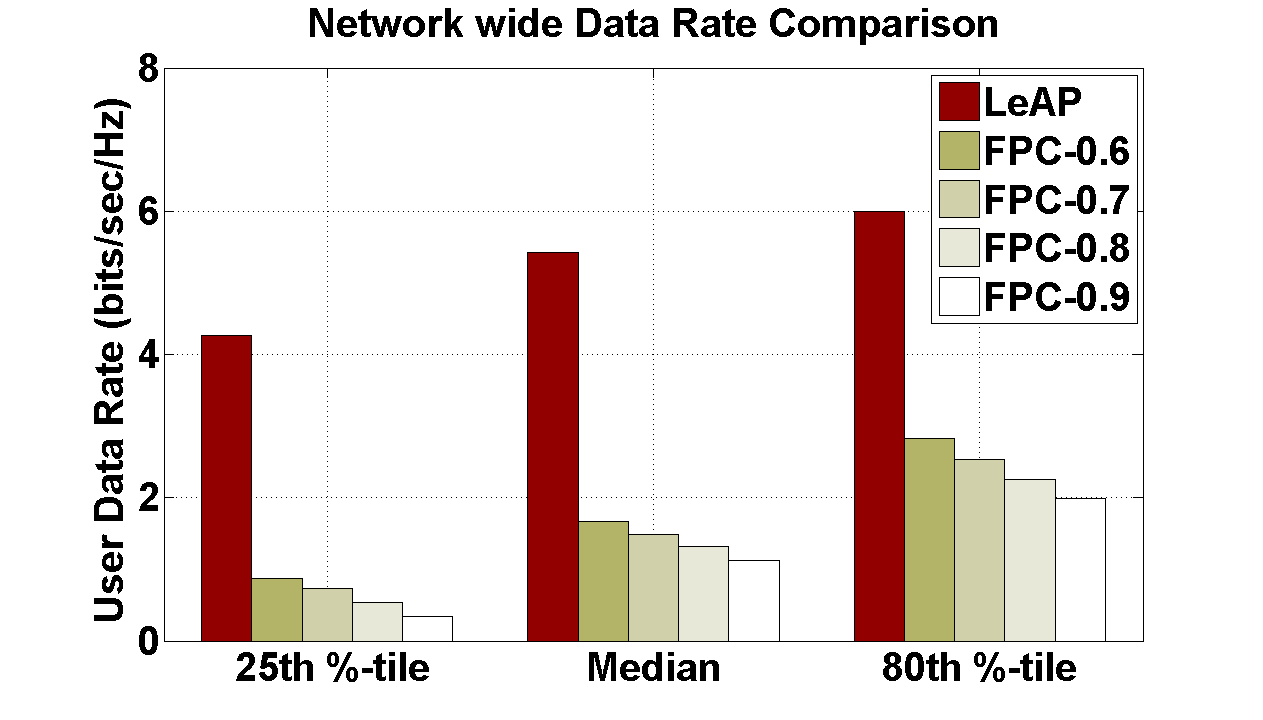}
\vspace{-0.1in}
\label{fig:allperf}}
\subfigure[]{
\includegraphics[height=1.6in, width=3.2in]{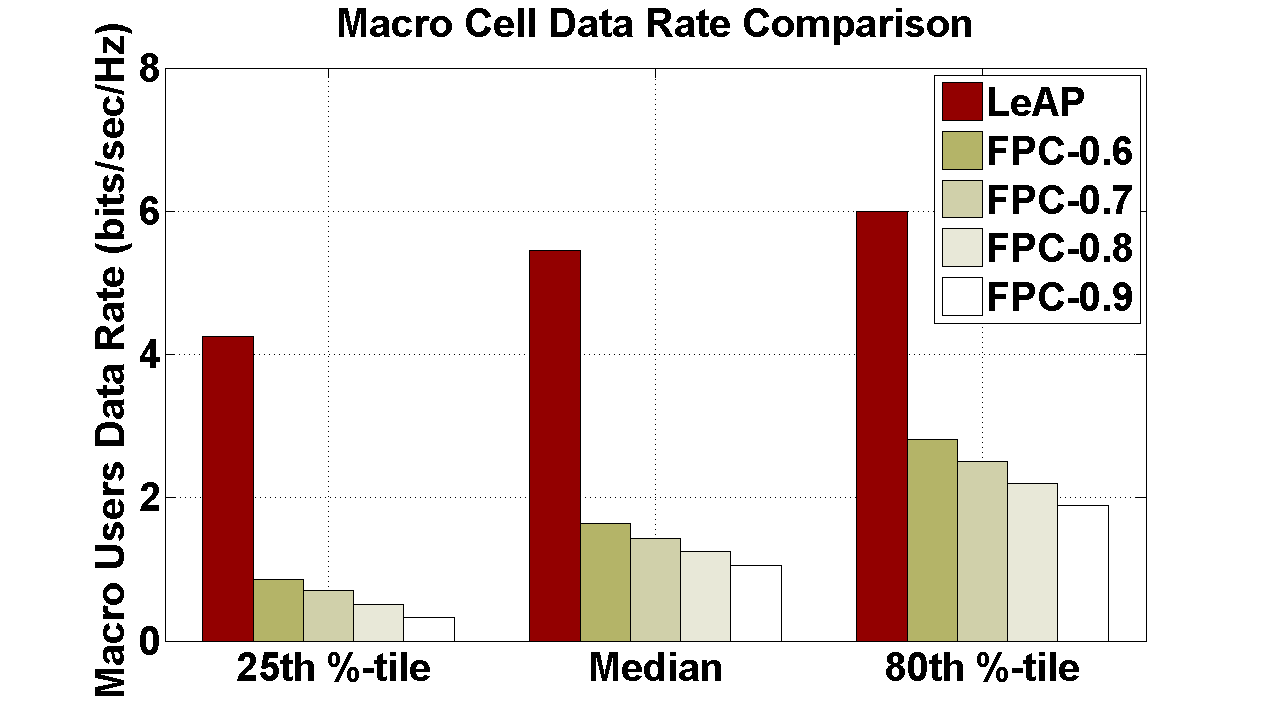}
\vspace{-0.1in}
\label{fig:macroperf}}
\subfigure[]{
\includegraphics[height=1.6in, width=3.2in]{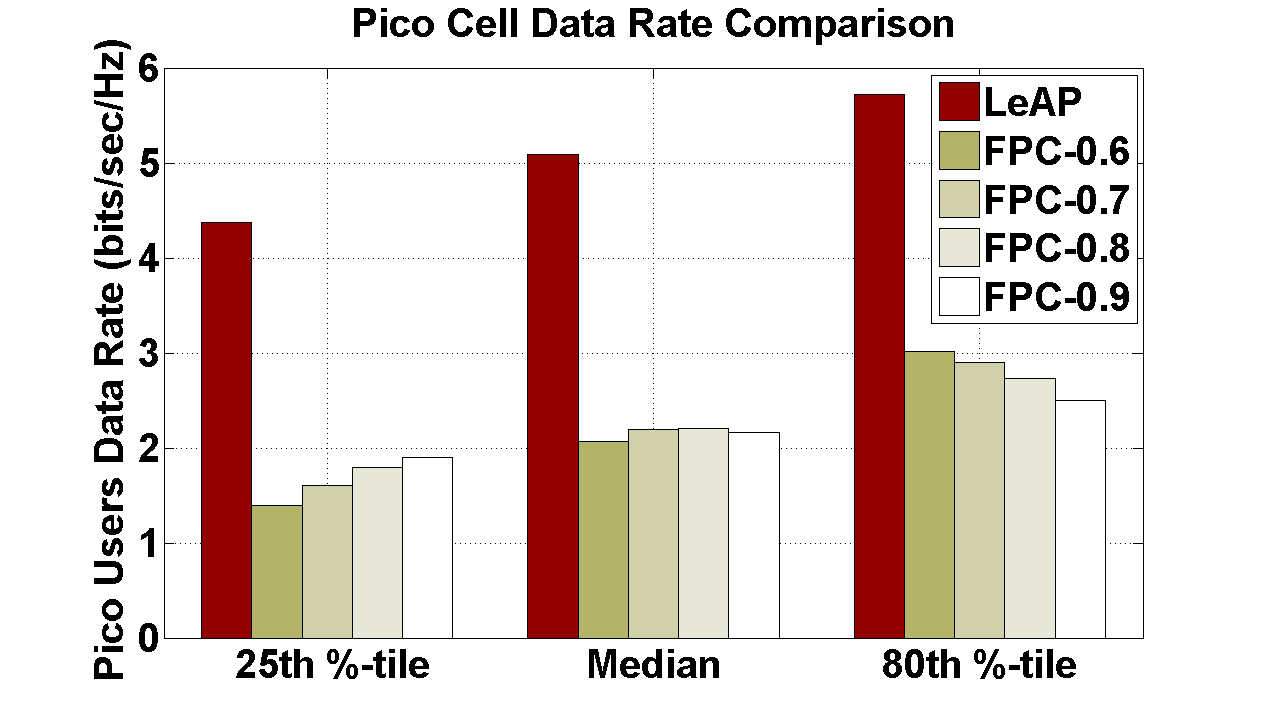}
\vspace{-0.1in}
\label{fig:picoperf}}
\vspace{-0.25in}
\caption {\small LeAP performance gain over fixed-$\alpha$ FPC for (a) all UEs, (b) macro cell UEs, 
and (c) pico cell UEs.\label{fig:leapperf}}
\vspace{-0.25in}
\end{center}
\end{figure}

{\bf Comparative schemes and LeAP parameters:} We compared the following schemes for our evaluation:

\noindent
1) {\bf LeAP Algorithm:} For most of our evaluation, we use
Algorithm~\ref{algo:iotspd} which has provable guarantees at a small expense of
computation time. Algorithm~\ref{algo:iotcce} requires expensive off-the-shelf
non-linear solver for large scale networks. Nevertheless, we show the performance of
Algorithm~\ref{algo:iotcce} using a free solver based stand-alone implementation.

\noindent
2) {\bf Fixed-${\boldsymbol \alpha}$ FPC (FA-FPC):} We compare the performance of LeAP with the following
approach popular in literature~\cite{fpcericc,castfpc08}: fix the same value of $\alpha$ (typically $\alpha=0.8$) for all
cells and set $P^{(0)}$ so that every mobile's SINR is above the decoding threshold.
SINR computation is performed using a nominal value of interference $I_{nominal}$
that is typically $5\--15$~dB above the noise floor. 

There are two important design parameters for LeAP: the histogram bin-size and the IoT-cap,
that is partly dictated by hardware designs constraints.  In our
evaluation, default choice of bin-size is 1~dB and the default IoT-cap is 20~dB. In
our problem framework (see IoTC-SCP), we set
$$I_{max}=\text{IoT-cap}-2.\text{bin-size}$$ as a design choice to provide cushion for
the path loss error introduced by histogram binning. We also show results by varying
the bin-size and IoT-cap. Also, $P_{max}=100$~mW/RB.

\vspace{-0.1in}
\subsection{Results}

{\bf LeAP Performance Gains:} In this section we evaluate the performance of
LeAP at various quantiles of the user data rate CDF produced across the
population of the users in our deployment area. Figure~\ref{fig:leapperf}
shows the LeAP performance with FA-FPC for all users, users that are associated
with macro cells, and users that are associated with pico cells. 
Our main observations are as follows:

\begin{sloppypar}
\begin{list}
{\labelitemi}
{\itemsep=0pt \parskip=0pt \parsep=0pt \topsep=0pt \leftmargin=0.2in}
\item As seen in Figure~\ref{fig:allperf} LeAP provides a data rate improvement over the best FA-FPC scheme, of $4.9\times, 3.25\times, 2.12\times$
for the $20^{th},50^{th},80^{th}$~\%-tile respectively. In other words, with LeAP, at least half the
users see a data rate improvement at least $3\times$ and 80\% of the users see an
improvement of at least $2\times$. The larger gain at lower percentiles indicate that
LeAP is really beneficial to users towards the cell edge who are most affected by
inter-cell interference.

\item As seen in Figure~\ref{fig:macroperf},~~\ref{fig:picoperf}, LeAP provides
higher improvement to macro users as compared to pico users. The median improvement
of pico users' data rate is $2.3\times$ whereas it is $3.3\times$ for macro users.
This can be explained by the fact that the macros use high power and have larger coverage 
leading to more users affected by inter-cell interference and thus more users
can reap the benefits of LeAP.

\end{list}
\end{sloppypar}

\begin{figure}[t]
\begin{center}
\includegraphics[height=1.6in, width=3.05in]{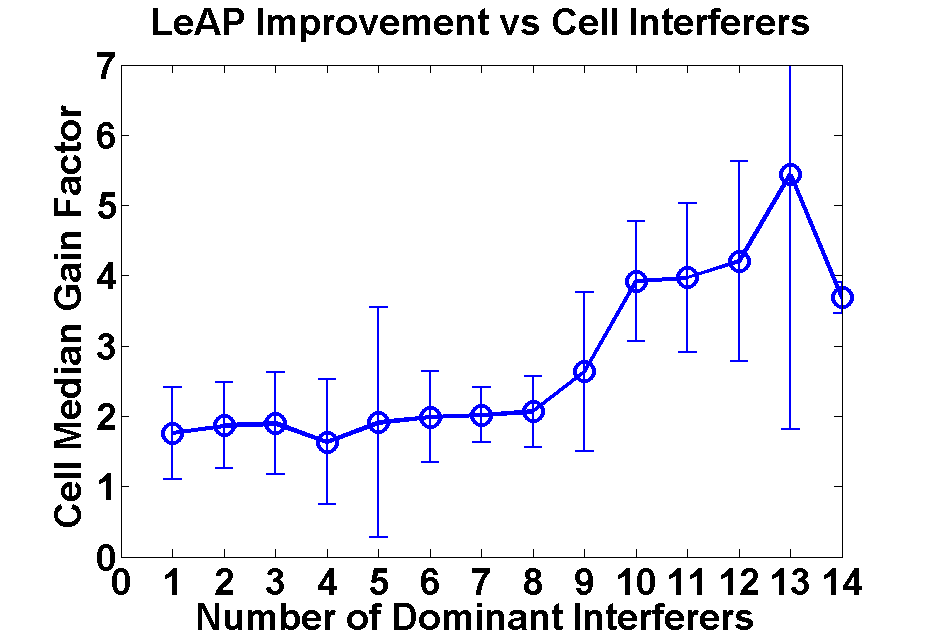}
\vspace{-0.15in}
\caption {\small Cell-wise LeAP gain vs. number of interferers.
\label{fig:cellint}}
\vspace{-0.35in}
\end{center}
\end{figure}

\begin{figure*}[t]
\begin{center}
\subfigure[]{
\includegraphics[height=1.6in, width=3.25in]{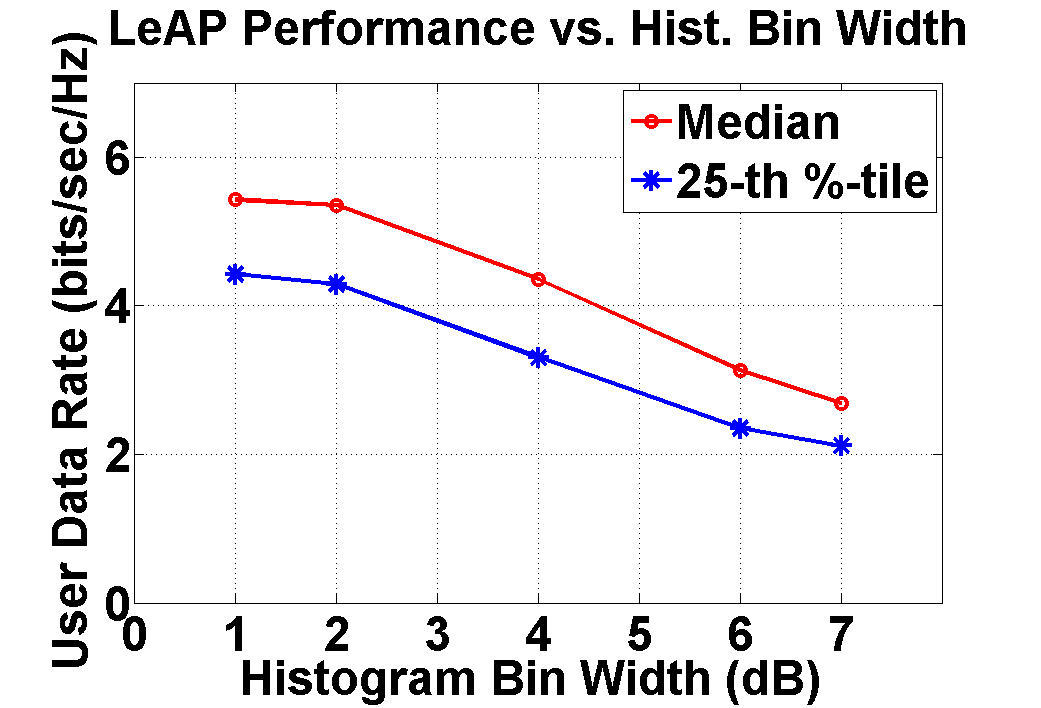}
\label{fig:binperf}}
\subfigure[]{
\includegraphics[height=1.6in, width=3.25in]{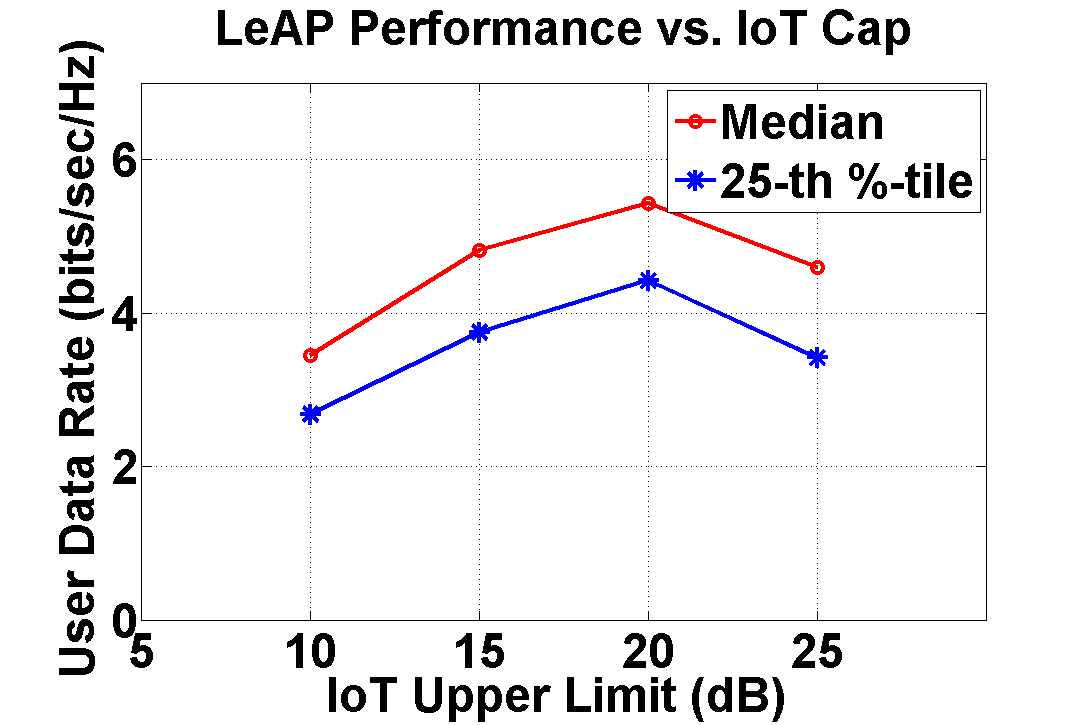}
\label{fig:iotperf}}
\vspace{-0.25in}
\caption {\small
Effect of (a) histogram bin size on median LeAP performance, 
and (b) effect of IoT-cap on median LeAP performance.
\label{fig:paramvar}
}
\vspace{-0.25in}
\end{center}
\end{figure*}

In Figure~\ref{fig:cellint}, we show the gains of different cells in
the network by grouping all the 125 cells based on number of dominant interferers. An
interfering cell is dominant at least $5\%$ of its uplink users, the users that
selected this cell as serving, that can interfere, i.e.,
$\Pr(\text{interference})\geq 0.05.$.  Figure~\ref{fig:cellint} shows that the
typical median data rate gain of cell-groups with 8~or~fewer interferers is around
$2\times$ whereas, the gains are around $4\times$ for cell-groups with more than
10~interferers. We also show error bar for each cell group based on the sample
standard deviation of gains within each group. We remark that the plot shows
irregularity at 13 and 14 interferers due to having just 3 cells with 13\--14
interferers.  These results suggest LeAP and other SON algorithms need not be applied
uniformly across the network but only in the {\em interference} problematic regions
of the network. Identifying these regions from measurement data is an interesting
problem outside the scope of this paper. 

{\bf LeAP performance with variation of histogram bin-size and IoT-cap:} In
Figure~\ref{fig:paramvar}, we show the median performance gains of LeAP by varying
histogram bin-size, with IoT-cap either fixed at 20~dB or varying (with
bin-size fixed at 2~dB) . Our main observations are as follows:
\begin{list}
{\labelitemi}
{\itemsep=0pt \parskip=0pt \parsep=0pt \topsep=0pt \leftmargin=0.2in}
\item {\em Effect of bin-size:} From Figure~\ref{fig:binperf} 
we see that the performance of LeAP and the gains deteriorate
with larger histogram bin-size which is intuitive. However, there is marginal
performance loss for 2~dB bin-size as compared to 1~dB bin-size, but the
performance deteriorates considerably for 6~dB bin-size. This suggests that the
right bin-size should not exceed $2$~dB for optimal LeAP performance.

\item {\em Effect of IoT-cap:} The performance of LeAP with IoT-cap is shown in
Figure~\ref{fig:iotperf}. Till an IoT-cap of around 20~dB,
the median performance improves with increase in $I_{max}$ as there is greater
flexibility in optimization. Past that, the gain decreases as can be seen with an IoT-cap of
25~dB.  This behavior is intriguing and can be explained as follows.  Let $I_{actual}$ be the actual 
interference in the network for given user
location. Also, $I_{max}$ is the maximum interference in IoT-SCP problem formulation.
Note that, $I_{actual}$ could  be higher than $I_{max}$ due to approximations
introduced by histogram binning. Intuitively, with larger $I_{max}$, the gap between
$I_{actual}$ and $I_{max}$ is larger as there are more interfering cells with
non-negligible interference signal. Thus, beyond a point, the binning approximations
somewhat nullify the benefits of increase in $I_{max}$.

\end{list}

{\bf LeAP computation times:} For the 125~cell network, Algorithm~\ref{algo:iotspd}
showed excellent gains after running for $50,000$ iterations.  Using a
3.4~GHz~Intel Xeon CPU with 16 GB RAM on a 6~core machine with 64~bit OS, the
running times of our multi-threaded implementation were on an average {\bf 12~sec}
with a variability of around 5~sec. Since we expect the periodicity of measurement
reports to be in minutes, this shows the scalability of our design.

{\bf Comparison of IoTC-CE and IoTC-SL:} We also have developed a stand-alone
implementation of IoTC-CE algorithm using a open source NLP solver called
CVX~\cite{cvx}.  CVX is a powerful tool for moderate sized problems that can deal
with the constraints in the IoTC-CE problem, that involve log-sum-exponential
functions.  However, CVX cannot operate for 125 cells in our evaluation network due
to associated variable limits in the solver, and so we choose a sub-region consisting
of 35 cells including 4~pico cells. Furthermore, since CVX only takes objective
functions of certain forms, we modified our objective by choosing utility as
$V(\gamma)=\ln(\sinr_{\min}+\ln(\exp(\gamma)))$, clearly an approximation.  

In Figure~\ref{fig:ceslcomp}, we compare the median data rate of the two LeAP
algorithms and the best FPC scheme in our evaluation. IoTC-CE still provides $90\%$
improvement for all cells and around $2\times$ improvements for macro cells. For pico
cells, IoTC-CE outperforms IoTC-SL. Also, the performance of median data rate of
IoTC-CE is $36\%$ less compared to IoTC-SL. This could be partially due to that
the approximated $V(.)$ worsens at lower SINR.

\begin{figure}[t]
\begin{center}
\includegraphics[height=1.6in, width=3.1in]{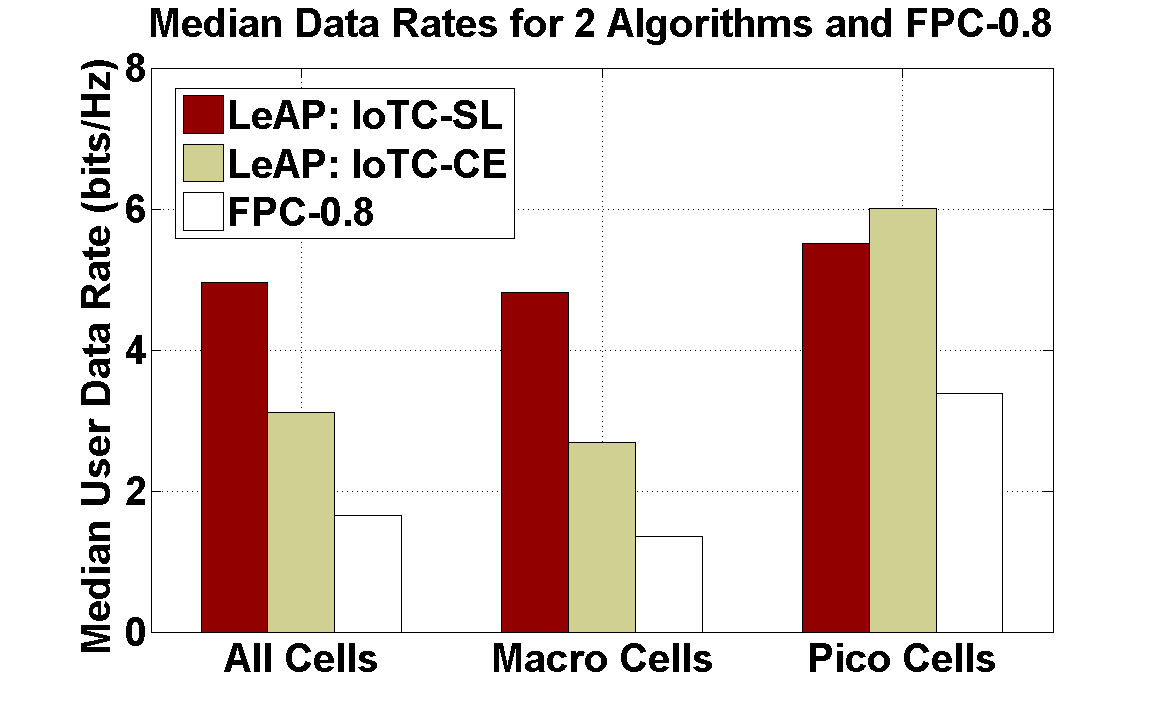}
\vspace{-0.15in}
\caption {\small Median data rate comparison for IoTC-CE and IoTC-SL. Histogram bin width is
$0.5$~dB. \label{fig:ceslcomp}}
\vspace{-0.25in}
\end{center}
\end{figure}

\vspace{-0.2in}
\section{Concluding Remarks}

In this work we have proposed LeAP, a measurement data driven approach for optimizing
uplink performance in an LTE network. Using propagation data from a real operational
LTE network, we show the huge gains to be had using our approach.

\vspace{-0.1in}
{\scriptsize
\bibliographystyle{acm}
\bibliography{myref}
}

\appendices

\section{Two Simple Useful Lemmas}

We will need the following inequalities in our proof. The proofs of the inequalities
are there in a longer version of the paper.

\begin{lemma}
\label{LEM:TECHLEM1}
The following holds for $a_0=1$ and real $a_i >0,i=1,2,\hdots,k$:
$$\ln^2(\sum_{i=0}^k e^{a_i})\leq \ln^2(k+1) + (1+2\ln(k+1))\sum_{i=0}^k a_i^2\ .$$
\end{lemma}

\begin{lemma}
\label{LEM:TECHLEM2}
For real positive $x_1,x_2$, the following holds:
\[ \ln^2(x_1+x_2)\leq \ln^2(e+x_1+x_2)+\ln^2 x_2\ .\]
\end{lemma}
\section{Proof of Theorem~\ref{THM:SPD}}

We will use results from stochastic approximation theory to prove our results.
Broadly speaking, stochastic approximation theory states that, under suitable
conditions, discrete iterations with random variables converge to an equilibrium
point of an equivalent deterministic ODE (ordinary differential equation)~\cite{borkarstochapp}. In other
words, to show that a discrete-stochastic update rule converges to a desired
point, we need to show the following steps~\cite{borkarstochapp}: {\em (i)} define an equivalent deterministic
ODE, {\em (ii)} argue that the ODE has an equilibrium point identical to to the
underlying discrete-stochastic update rule, and {\em (iii)} argue that conditions for
limiting behavior of discrete-stochastic update rule to be similar to that of the ODE
is satisfied (these conditions are provided by stochastic approximation theory).

{\bf Defining a limiting ODE:} 
Let $m$ be the dimension of ${\mathbf z}$ or equivalently the number of primal
variables. First, define the following operator $\Gamma_{{\mathbf z}}({\mathbf y})$ parametrized by
primal variable ${\mathbf z}$:
\begin{equation*}
\Gamma_{{\mathbf z}}({\mathbf y})
=\left[ 
\Gamma_{z_1}^{(1)}(y_1)\ \ 
\Gamma_{z_2}^{(2)}(y_2)\hdots
\Gamma_{z_m}^{(m)}(y_m)\ \ 
\hdots
\right]^t\ ,
\end{equation*}
where
\begin{equation*}
\Gamma_{z_i}^{(i)}(y_i)
=\left\{
\begin{array}{cl}
y_i\ind{y_i>0} &,\ \text{if}\ z_i=z_{i,\min} \\
y_i\ind{y_i<0} &,\ \text{if}\ z_i=z_{i,\max} \\
y_i &,\ \text{else}
\end{array}
\right.
\end{equation*}
In the above, $z_{i,\min}$ and $z_{i,\max}$ are the lower and the upper bounds of the
corresponding primal variables. Similarly, define the functions $\Gamma_{{\mathbf p}}^+({\mathbf
q})$ parametrized by ${\mathbf p}=[p_1,\ p_2,\ \hdots p_n]$ as
\begin{align*}
\Gamma^+_{{\mathbf p}}({\mathbf q})&=
\left[ \Gamma_{p_1}^+(q_1)\ \ \Gamma_{p_2}^+(q_2)\ \hdots \ \Gamma_{p_n}^+(q_n)\right]
\end{align*}
where,
\begin{align*}
\Gamma_{p_j}^+(q_j)
&=\left\{
\begin{array}{cl}
q_j\ind{q_j>0} &,\ \text{if}\ p_j=0 \\
q_j &,\ \text{if}\ p_j>0 \\
\end{array}
\right.
\end{align*}

We now define the following equivalent ODE for the iterative update rule in
Algorithm~IoTC-SDP.
\begin{align}
\begin{bmatrix}
{\mathbf {\dot z}}(t) \\ 
{\mathbf {\dot p}}(t)
\end{bmatrix}
& = \begin{bmatrix}
\Gamma_{{\mathbf z}(t)}\left(\Ex[{\mL}_{{\mathbf z}}({\mathbf z}(t), 
h({\mathbf z}(t), {\mathbf O}, {\boldsymbol \Lambda}),{\mathbf p}(t))]\right) \nonumber \\
\Gamma_{{\mathbf p}(t)}^+(\Ex[h({\mathbf z}(t), {\mathbf O}, {\boldsymbol \Lambda})])
\end{bmatrix}\nonumber \\
& \label{eqn:ode}
\end{align}

{\bf Asymptotic equivalence of ODE:} 
The fundamental result in stochastic approximation theory~\cite{borkarstochapp} implies that
the ODE defined in~(\ref{eqn:ode}) has similar asymptotic behavior to the updates
defined in Algorithm~IoTC-SL provided certain conditions hold. The following lemma
establishes this equivalence by verifying these conditions in our setup.

\begin{lemma}
\label{LEM:SA-SPD}
The stochastic primal-dual update rule of
Algorithm~IoTC-SL given
by~(\ref{eqn:pud})~and~(\ref{eqn:dud}) converges to the invariant set of ODE given
by~(\ref{eqn:ode}).
\end{lemma}
\begin{proof}
We introduce some notations to start with.  Rewrite the iterative update
steps~(\ref{eqn:pud})~and~(\ref{eqn:dud}) in Algorithm~IoTC-SL in a compact form as
\begin{align*}
\begin{bmatrix} \bz_{n+1} \\ \bp_{n+1} \end{bmatrix}
& =  \mP\left(\begin{bmatrix} \bz_{n} \\ \bp_{n} \end{bmatrix}
+ a_n M_n \right)
\end{align*}
where,
\begin{equation}
M_n=
\begin{bmatrix}
\nabla_{\bz}\mL({\mathbf  z}_n, 
h({\mathbf z}_n, {\boldsymbol \chi}, {\boldsymbol \xi}),{\mathbf p}_n )\\
h({\mathbf z}_n, {\boldsymbol \chi}, {\boldsymbol \xi})
\end{bmatrix}
\end{equation}
and $\mP(.)$ denotes the projection of the primal variables into the compact
set defined by bounds in Lemma~\ref{LEM:PUB} and dual variables into the positive
axis.

According to Chapter~5.4~\cite{borkarstochapp}, the candidate equivalent ODE is given
by

\begin{align*}
\begin{bmatrix}
{\mathbf {\dot z}}(t) \\ 
{\mathbf {\dot p}}(t)
\end{bmatrix}
& = \Gamma_{(\bz(t),\bp(t))}\begin{bmatrix}
\left(\Ex[{\mL}_{{\mathbf z}}({\mathbf z}(t), 
h({\mathbf z}(t), {\mathbf O}, {\boldsymbol \Lambda}),{\mathbf p}(t))]\right) \nonumber \\
\Ex[h({\mathbf z}(t), {\mathbf O}, {\boldsymbol \Lambda})]
\end{bmatrix}\nonumber \\
\end{align*}
where
\[\Gamma_{(\bx)}(\by)=\lim_{\delta \downarrow 0}
\frac{\mP(\bx +\delta \by)-\mP({\by})}{\delta}\ .\]
One can verify that this candidate equivalent ODE is precisely the one defined
by~(\ref{eqn:ode}). Also, with a step-size of $a_n=1/n$,  for iterative update
steps~(\ref{eqn:pud})~and~(\ref{eqn:dud}) in Algorithm~IoTC-SL to converge to an
invariant set of the candidate ODE, we need to verify 
two conditions for our case (Chapter~5.4~\cite{borkarstochapp}):
\begin{itemize}

\item {\bf C1:} The functions, $\Gamma_{{\mathbf z}(t)}\left(\Ex[{\mL}_{{\mathbf z}}({\mathbf z}, 
h({\mathbf z}, {\mathbf O}, {\boldsymbol \Lambda}),{\mathbf p})]\right) $ and
$\Gamma_{{\mathbf p}}^+(\Ex[h({\mathbf z}, {\mathbf O}, {\boldsymbol \Lambda}))$ are 
Lipschitz\footnote{\scriptsize A map $\bx \rightarrow F(\bx)$ is called Lipschitz if
$\norm{F(\bx_1)-F(\bx_2)}\leq K\norm{\bx_1 - \bx_2}$ for scalar constant
$K\in (0,\infty)$.} in ${\mathbf z}$ and ${\mathbf p}$.

\item {\bf C2:} The random variable $M_n$ satisfies

$$\Ex[\normsq{M_n}]\leq K(1+\normsq{{\mathbf z}_n}+\normsq{{\mathbf p}_n})$$
for some constant $K>0$.

\end{itemize}

The conditions are not difficult to verify for our problem. We outline the
steps in the following.

{\em Verifying} {\bf C1:} We need to show the Lipschitz continuity of the r.h.s.
of~(\ref{eqn:ode}) in ${\mathbf z}$ and ${\mathbf p}$. First note that the
ODE~(\ref{eqn:ode}) is defined such that ${\mathbf z(t)},t\geq 0$ lies within the
compact set defined in Lemma~\ref{LEM:PUB} (provided the initial conditions are also
within that set) from which it follows that all terms (in the r.h.s. of ~(\ref{eqn:ode})) 
containing the primal variables
are bounded. Also, the r.h.s of~(\ref{eqn:ode}) is a linear function of the dual
variables ${\mathbf p}$. One can combine these two facts and argue easily l.h.s.
of~(\ref{eqn:ode}) is Lipschitz in ${\mathbf z}$ and ${\mathbf p}$. We skip the
details.

{\em Verifying} {\bf C2:} Note from~(\ref{eqn:compcons}) that, the randomness in
$h({\mathbf z}_n, {\boldsymbol \chi}, {\boldsymbol \xi})$ is only due to randomness
in $Z_{c,n}$ where
$$Z_{c,n}=\ln \left( \sum_{e\in \mJ_c} \chi_e e^{(\pi_{e,n} +\alpha_{e,n} \xi_e -
\xi_{e\rightarrow c})} + N_0 \right).$$
 We can use the boundedness of the primal variables ${\mathbf z}$ to show that 
\begin{equation}
\label{eqn:mn}
\Ex[M_n^2]\leq K_1 + \sum_{c\in\mC}K_2\Ex[Z_{c,n}^2]\ ,
\end{equation}
for suitable positive constants $K_1$ and $K_2$. 
To see this, first observe the following.
\begin{align*}
&\Ex[Z_{c,n}^2]\\
&\leq\Ex\left[\ln^2 \left( e + \sum_{e\in \mJ_c} \chi_e e^{(\pi_{e,n} +\alpha_{e,n} \xi_e -
\xi_{e\rightarrow c})} + N_0 \right)\right]+\ln^2 N_0 
\end{align*}
where we have applied Lemma~\ref{LEM:TECHLEM2}.
Define $A>1$ as the upper bound on the path loss in linear-scale in the entire network.
Now notice that the function $f(x)=\ln^2(x+e)$ is concave for $x>0$. By
applying Jensen's inequality and definition of $A$,
\begin{align*}
&\Ex[Z_{c,n}^2]\\
&\leq\ln^2\left(\sum_{e\in \mJ_c} \Ex\left[\chi_e e^{(\pi_{e,n} +\alpha_{e,n} \xi_e -
\xi_{e\rightarrow c})}\right] + e + N_0 \right) +\ln^2 N_0 \\ 
&\leq\ln^2\left(A\sum_{e\in\mJ_c}e^{\pi_{e,n}} + e + N_0 \right)+\ln^2 N_0 \\ 
\end{align*}
Applying Lemma~\ref{LEM:TECHLEM1} to the previous step, we have for suitable constants 
$K_3$ and $K_4$ (i.e., $K_3,\ K_4$ are independent
of $\pi_{e,n}$'s),
\begin{align*}
\Ex[Z_{c,n}^2]
&\leq K_3 + K_4\sum_{e\in\mJ_c} \ln^2(e^{\pi_{e,n}})
= K_3 + K_4 \sum_{e\in \mJ_c}\pi^2_{e,n}
\end{align*}
The above combined with~(\ref{eqn:mn}) implies that,
\[\Ex[M_n^2]\leq K_5 + K_6 \normsq{{\boldsymbol \pi_n}}\]
for system dependent constants $K_5,K_6$ (i.e., $K_5,\ K_6$ could depend on the
network topology, but they are independent of the primal and dual variables). Condition
{\bf C2} is thus verified.

The lemma is proved since we have verified {\bf C1} and {\bf C2}.
\end{proof}

{\bf Global Stability of the ODE:}
Now that we have argued that Algorithm~IoTC-SL behaves like the ODE~(\ref{eqn:ode})
asymptotically, we prove that ODE~(\ref{eqn:ode}) converges to the optimal point
of problem~IoTC-SCP. 

\begin{lemma}
\label{LEM:STABLE}
The ODE given by~(\ref{eqn:ode}) is globally asymptotically stable and the system
converges to saddle point solution of $\mL({\mathbf z}, {\mathbf p})$ which also corresponds to
the optimai value of ${\mathbf z}$.
\end{lemma}
\begin{proof}
Consider the Lyapunov function 
\[W(t)
=\normsq{\bz(t)-\bz^*} + \normsq{\bp(t)-\bp^*}\ ,
\]
where, ${\bz}^*$ and ${\bp}^*$ are the primal and dual optimal solution of the
problem IoTC-SCP. Now note the following.

\begin{align}
\dot{W}(t)&=
2(\bz(t)-\bz^*)^t \dot{\bz}(t) + 2(\bp(t)-\bp^*)^t \dot{\bp}(t)
\nonumber \\
&= 2(\bz(t)-\bz^*)^t 
\Gamma_{{\mathbf z}(t)}\left(\Ex[{\mL}_{{\mathbf z}}({\mathbf z}(t), 
h({\mathbf z}(t), {\mathbf O}, {\boldsymbol \Lambda}),{\mathbf p}(t))]\right)
\nonumber \\
& \ \ + 2(\bp(t)-\bp^*)^t 
\Gamma_{{\mathbf p}(t)}^+(\Ex[h({\mathbf z}(t), {\mathbf O}, {\boldsymbol \Lambda})])
\label{eqn:wdot1}
\end{align}
Now note from the definition of $\Gamma_{\bz}(.)$ that,
\[(z_{i}-z_i^*)\Gamma_{z_i}^{(i)}(y_i)\leq (z_{i}-z_i^*)y_i \ ,\]
since $z_{i,\min}\leq z_i^* \leq z_{i,\max}$. Also,
\[(p_{i}-p_i^*)\Gamma_{p_i}^+(q_i)\leq (p_{i}-p_i^*)q_i \ .\]
It thus follows from~(\ref{eqn:wdot1}) that
\begin{align}
&\dot{W}(t)
\leq 2(\bz(t)-\bz^*)^t 
\Ex[\nabla_{\bz}{\mL}({\mathbf z}(t), 
h({\mathbf z}(t), {\mathbf O}, {\boldsymbol \Lambda}),{\mathbf p}(t))]
\nonumber \\
& \ \ - 2(\bp(t)-\bp^*)^t 
\Ex[\nabla_{\bp}\mL({\mathbf z}(t), 
h({\mathbf z}(t), {\mathbf O}, {\boldsymbol \Lambda}),{\mathbf p}(t))]
\label{eqn:wdot2}
\end{align}
For notational convenience, let $G(.)$ denote the function
$$G(\bz,\bp)=\Ex[{\mL}({\mathbf z}, 
h({\mathbf z}, {\mathbf O}, {\boldsymbol \Lambda}),{\mathbf p})]\ .$$ 
Note that $G(\bz,\bp)$ is concave in 
$\bz$ and convex in $\bp$. Thus, from the proerty of convex and concave
functions\cite{convopt}
\[(\bz - \bz^*)^t \nabla_{\bz}G(\bz,\bp)\leq G(\bz,\bp)-G(\bz^*,\bp)\]
and
\[(\bp - \bp^*)^t \nabla_{\bp}G(\bz,\bp)\geq G(\bz,\bp)-G(\bz,\bp^*)\ ,\]
which applied to~(\ref{eqn:wdot2}) implies
\begin{align}
&\dot{W}(t) 
\nonumber\\
&\leq G(\bz(t),\bp(t))-G(\bz^*,\bp(t)) - (G(\bz(t),\bp(t))-G(\bz(t),\bp^*))
\nonumber\\
&= G(\bz(t),\bp^*)-G(\bz^*,\bp(t))
\label{eqn:wdot3}
\end{align}
Since $(\bz^*,\bp^*)$ is a saddle point solution of 
$$\max_{\bz}\min_{\bp\geq 0}G(\bz,\bp)=\min_{\bp\geq 0}\max_{\bz}G(\bz,\bp)\ ,$$
we have,
$ G(\bz,\bp^*) \leq G(\bz^*,\bp^*) \leq G(\bz^*,\bp)$
for any $\bz$ and $\bp\geq 0$, 
from which it follows that
\[\dot{W}(t)\leq 0\ ,\]
with equality iff $\bz(t)=\bz^*$ and $\bp(t)=\bp^*$.
Thus, it follows from La'Salle's invariance principle~\cite{khalilnls} that the ODE given
by~(\ref{eqn:ode}) converges to $(\bz^*,\bp^*)$.

\end{proof}

{\bf Putting it all together: Proof of Theorem~\ref{THM:SPD}.} Lemma~\ref{LEM:SA-SPD}
shows that Algorithm~IoTC-SL converges to the invaraint set of ODE~(\ref{eqn:ode}) and
Lemma~\ref{LEM:STABLE} shows that the ODE converges to the optimal point of problem
IoTC-SCP. Thus Algorithm~IoTC-SL converges to the optimal point of problem
IoTC-SCP. Hence the proof.


\end{document}